\documentclass[journal]{IEEEtran}
\usepackage{lineno}
\usepackage{cite}
\usepackage{hyperref}
\usepackage[tbtags]{amsmath}
\usepackage{amssymb,amsfonts}
\usepackage{amsthm}
\usepackage{amsbsy}

\usepackage{graphicx}
\usepackage{textcomp}
\usepackage{xcolor}
\usepackage{float}
\usepackage{mathrsfs}
\usepackage{mathtools}
\usepackage{multirow}
\usepackage{algpseudocode}
\usepackage{comment}
\usepackage[inline, shortlabels]{enumitem}
\makeatletter
\newif\if@restonecol

\usepackage{algorithm}

\usepackage{caption}
\usepackage{subcaption}
\usepackage{ragged2e}

\usepackage{dsfont}
\usepackage{bbm}
\usepackage{stfloats}
\usepackage{bm}
\usepackage{orcidlink}

\newtheorem{corollary}{Corollary}

\theoremstyle{definition}
\newtheorem{theorem}{Theorem}

\newtheorem{lemma}{Lemma}

\newcommand{\biggg}{\bBigg@{3}}
\newcommand{\Biggg}{\bBigg@{3.5}}
\makeatother
\begin{document}

\title{ Sensing-Aided Secure Multicast in Rotatable Antenna-Enabled ISAC Systems}

\author{Zequan~Wang, Liang~Yin, Chongjun~Ouyang, Hao~Xu, Yunan~Sun, Yitong~Liu, and Hongwen~Yang

\thanks{Zequan~Wang, Liang~Yin, Hao~Xu, Yunan~Sun, Yitong~Liu, and Hongwen~Yang are with the School of Information and Communication Engineering, Beijing University of Posts and Telecommunications, Beijing, 100876, China (e-mail: \{zequanwang, YinL, Xu\_Hao, sunyunan, liuyitong, yanghong\}@bupt.edu.cn).

Chongjun~Ouyang is with the School of Electronic Engineering and Computer Science, Queen Mary University of London, E1 4NS London, U.K. (e-mail: c.ouyang@qmul.ac.uk).}

\thanks{(Corresponding author: Liang~Yin)}
}

\maketitle

\begin{abstract}
Acquiring the channel state information (CSI) of passive
eavesdroppers remains a fundamental challenge in physical layer
security. The sensing capability of integrated sensing and
communication (ISAC) systems enables estimation of a potential
eavesdropper's angle of departure (AoD) before secure transmission.
Accordingly, a sensing-aided secure multicast scheme is
proposed using a rotatable antenna (RA) architecture that combines
array-level and element-level rotations with analog beamforming. 
The scheme comprises eavesdropper sensing and secure communication stages.
In the sensing stage, the maximum likelihood estimator
(MLE) and corresponding Cram\'er--Rao bound (CRB) are derived for
eavesdropper AoD estimation. The two rotation levels are then
optimized through cyclic coordinate search to minimize the worst-case
CRB. The resulting AoD estimate and CRB determine the center and width of the angular
uncertainty region, respectively.
In the communication stage, the constant-modulus analog beamformer and RA
configuration are jointly optimized to maximize the worst-case secrecy
rate over this region. After angular discretization and smooth
approximation, the resulting problem is solved using a product-space
joint optimization framework. 
Numerical simulation results validate the
convergence and effectiveness of the proposed algorithms. It is
demonstrated that i) the proposed RA-enabled sensing design effectively improves
the eavesdropper AoD estimation accuracy; ii) a high and nearly constant secrecy rate is
maintained over the uncertainty region;
and iii) the joint optimization of the two rotation levels
yields lower CRBs and higher secrecy rates than schemes employing
either a fixed-position array or a single rotation level.
\end{abstract}


\section{Introduction}
Wireless communication networks are inherently vulnerable to
information leakage because of the broadcast nature of wireless
transmission \cite{A2010}. By exploiting the randomness of wireless
channels, physical layer security (PLS) has emerged as an effective
complement to conventional cryptographic methods
\cite{V2021,L2017}. 
Various techniques have been developed to enhance
PLS. Among them, secure beamforming has been extensively investigated as a means of exploiting spatial degrees of freedom (DoFs)
to enhance legitimate transmission while suppressing information leakage to eavesdroppers \cite{Chen2016, Wu2018}.
However, many existing secure beamforming designs require the channel state information (CSI) of the eavesdropper, which is difficult to acquire because passive eavesdroppers do not normally transmit pilot signals.

The sensing capability of integrated sensing and communication (ISAC)
provides a means of alleviating this limitation
\cite{Matsumine2025,Su2023}. By sensing environmental targets, the base station (BS)
can infer the direction, distance, or location of a potential
eavesdropper, thereby obtaining useful side information for secure
transmission design. This has motivated sensing-aided secure
communication (SASC). 
Existing studies have investigated two-stage
sensing and secure transmission designs \cite{Cao2025} and sensing-assisted secure
beamforming \cite{Liu2025,Zhang20261,Zhang2026}. More recently, reconfigurable antennas have been introduced into SASC to provide additional spatial flexibility.
Specifically, the transmit and receive movable antenna (MA) positions were optimized in \cite{Chen2026} 
to improve the Cram\'er--Rao bound (CRB) of eavesdropper angle-of-departure (AoD) estimation, with the resulting angular uncertainty further incorporated into robust beamforming. 
These studies highlight that sensing-derived eavesdropper information is inherently imperfect and that the associated estimation uncertainty should be explicitly accounted for in subsequent secure transmission design.

Extending SASC to secure multicast transmission introduces a
different spatial design problem. A common confidential message must
be reliably delivered to multiple legitimate users distributed over
different directions, while information leakage must be suppressed
over the possible eavesdropper AoDs. Consequently, the secrecy rate
is jointly limited by the weakest legitimate user and the worst-case
eavesdropper direction. This problem becomes more restrictive under
analog beamforming because the constant-modulus constraint imposed by phase shifters limits the
available beamforming flexibility. Although analog beamforming provides a
hardware-efficient implementation for large antenna arrays
\cite{K2012,Pan2014,Zhang2018}, a fixed-position array (FPA) cannot
adapt its physical radiation characteristics to either the user
distribution or the eavesdropper angular uncertainty. 
These limitations have motivated the exploitation of physically reconfigurable antennas for secure multicast transmission, where additional spatial flexibility can be introduced through antenna reconfiguration.
A representative example is the MA, for which secure analog beamforming has been shown to improve multicast secrecy through antenna-position optimization \cite{Xiong2026}. 
Nevertheless, such gains generally rely on translational antenna movement, which may require additional movement space and mechanical control.
In contrast, rotatable antennas (RAs) provide orientation-domain reconfigurability by adjusting the antenna or array orientations without requiring translational repositioning of the individual antenna elements \cite{Zheng2026,Zheng20262}.
Such a compact form of physical reconfigurability is particularly attractive for secure multicast transmission, where the radiation characteristics need to be simultaneously adapted toward multiple legitimate users and uncertain eavesdropper directions. However, RA-enabled robust secure multicast transmission under sensing-induced angular uncertainty remains largely unexplored.

According to the rotation granularity, existing RA configurations can be broadly
classified into array-level and element-level designs
\cite{Zheng20262}. In an array-level design, the entire array rotates
as a whole, whereas an element-level design allows the boresight of
each antenna to be adjusted independently. The former provides common
aperture reorientation, while the latter offers finer control of the
orientation-dependent antenna gains.
Recent studies have investigated RA in wireless communication and
sensing. For communication systems, RA channel modeling and joint
boresight-beamforming optimization were developed in
\cite{Zheng20261}, channel estimation was studied in
\cite{XiongRA2025}, and RA-assisted secure transmission was considered
in \cite{Dai2025}. RA has also been applied to ISAC systems. 
Far-field RA-enabled ISAC was investigated in \cite{Zhou2025,Wang2026}, where array rotation was jointly designed with transmit beamforming.
Near-field RA-enabled ISAC and the associated range- and angle-estimation CRBs
were investigated in \cite{Zhang20262,Wang20261}.

Although these studies have demonstrated the benefits of antenna
rotation in both communication and sensing, most existing designs consider either a common array
rotation or independent element orientations. To more fully explore
the potential of RA, we jointly consider these two rotation levels,
which provide complementary aperture reorientation and
radiation-pattern control.
This joint design is well suited to maintaining multicast
coverage toward legitimate users while suppressing leakage over
uncertain eavesdropper directions.
However, jointly optimizing the array-level and element-level
rotations for the distinct objectives of sensing and secure
communication remains a challenging problem due to their nonlinear coupling.
In the sensing stage, RA configuration affects the eavesdropper AoD estimation accuracy
and hence the width of the resulting uncertainty region. In the
communication stage, RA configuration must be redesigned together
with the analog beamformer to suppress the worst-case leakage over
this region. Therefore, the angular uncertainty is not specified
\emph{a priori}, but is determined by the sensing accuracy achieved
under the sensing-stage RA configuration. 

Motivated by the above considerations, this paper investigates
sensing-aided secure multicast transmission in an RA-enabled ISAC
system with analog beamforming. Both array-level and element-level
rotations are considered to fully exploit the rotational flexibility
of RA. The sensing-stage RA configuration is optimized to improve
eavesdropper AoD estimation and construct a CRB-based angular
uncertainty region. The RA configuration is then redesigned together
with the analog beamformer to maximize the worst-case secrecy rate
over this region. The main contributions of this paper are summarized
as follows.

\begin{itemize}
    \item We develop an RA-enabled sensing-aided secure multicast framework
    that combines analog beamforming with array-level and element-level
    rotations. Separate RA configurations are adopted for sensing and
    communication to accommodate the different objectives of the two
    stages. In the sensing stage, the RA configuration is optimized to
    improve eavesdropper AoD estimation. In the communication stage, the RA configuration and
    analog beamformer are jointly optimized to maximize the worst-case
    secrecy rate over this region.

    \item We derive the maximum likelihood estimator (MLE) of the
    eavesdropper AoD and the corresponding CRB. The MLE determines the
    center of the angular uncertainty region, while the CRB evaluated
    at the MLE under the optimized sensing-stage RA configuration
    determines its width. The sensing-stage problem is equivalently
    reformulated as maximizing the worst-case effective angular
    information over the search region and is solved using a cyclic
    coordinate search algorithm.

    \item We formulate a robust max-min secrecy-rate problem for the
    communication stage by jointly optimizing the constant-modulus
    analog beamformer and both the array-level and element-level rotations. After angular
    discretization and smooth approximation, a product-space joint
    optimization (PSJO) framework is developed. The analog beamformer
    is characterized on a complex circle manifold, while the bounded
    rotation variables are handled by projection onto a
    box-constrained set. A conjugate-gradient method with Armijo line
    search is developed to jointly update these variables.

    \item Numerical results validate the convergence and effectiveness
    of the proposed algorithms. It is further shown that: (i) both
    algorithms exhibit monotonic convergence within a moderate number
    of iterations; (ii) the proposed sensing-stage RA optimization significantly improves eavesdropper AoD estimation accuracy;
    and (iii) the proposed sensing-aided secure transmission design achieves higher worst-case multicast secrecy rates than the benchmark schemes while maintaining robust performance against eavesdropper AoD-estimation errors.
\end{itemize}

Notation: Vectors and matrices are denoted by bold lowercase and uppercase letters, respectively. 
$(\cdot)^{*}$, $(\cdot)^{T}$, $(\cdot)^{H}$, $\|\cdot\|_2$, $|\cdot|$, and $\|\cdot\|_F$ denote the conjugate, transpose, conjugate transpose, Euclidean norm, absolute value, Frobenius matrix norm, respectively. $\mathbb{C}^{M \times N}$ and $\mathbb{R}^{M \times N}$ are the sets for complex and real matrices of $M \times N$ dimensions, respectively. $\mathbf{I}_N$ is the identity matrix of order $N$.
$\operatorname{Re}{(\cdot)}$ and $\operatorname{Im}{(\cdot)}$ denote the real and imaginary parts of a complex number.
\section{System Model}
\subsection{RA-BS Model}

\begin{figure*}[!t]
    \centering
        \begin{subfigure}[t]{0.49\linewidth}
        \centering
        \includegraphics[height=0.49\linewidth]{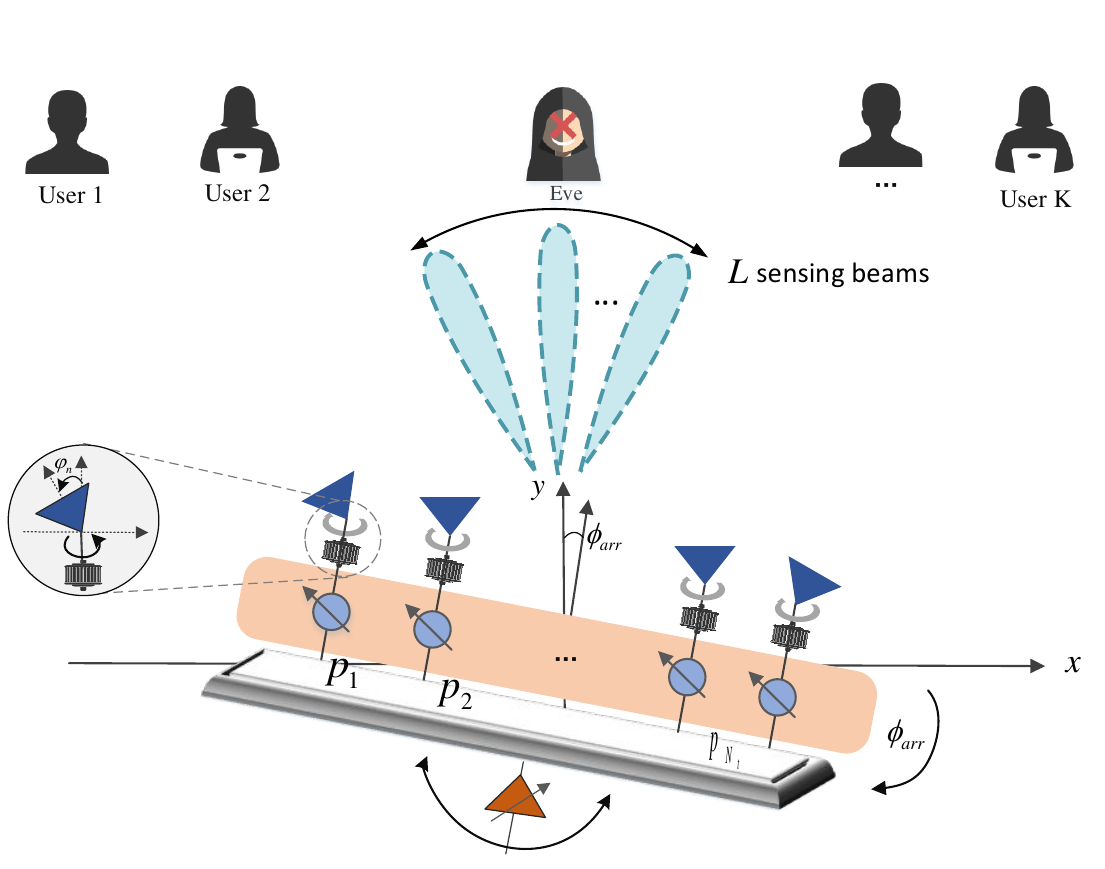}
        \caption{Eavesdropper sensing stage.}
        \label{fig:sensing_stage}
    \end{subfigure}
        \hfill
        \begin{subfigure}[t]{0.49\linewidth}
        \centering
        \includegraphics[height=0.49\linewidth]{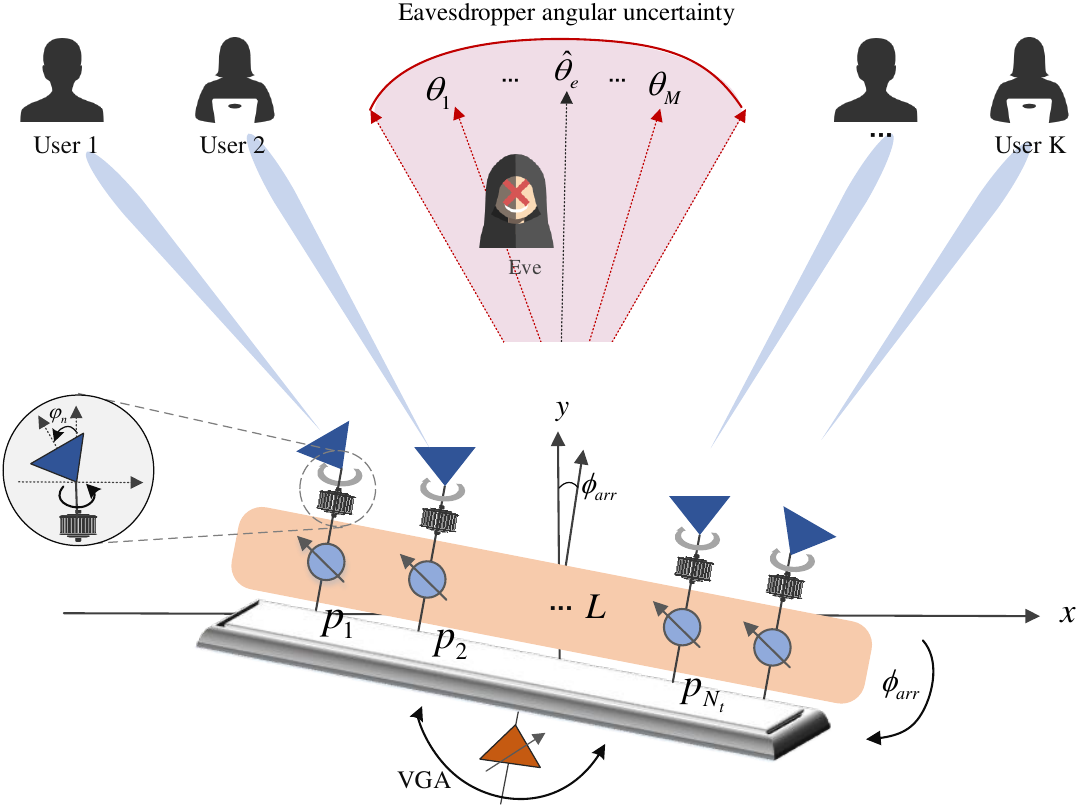}
        \caption{Secure communication stage.}
        \label{fig:communication_stage}
    \end{subfigure}
    \caption{\justifying Illustration of the proposed RA-aided
    sensing and secure multicast system.}
    \label{fig:system_model}
\end{figure*}
{
As illustrated in Fig.~\ref{fig:system_model}, we consider an RA-enabled multi-user multiple-input single-output (MU-MISO) secure multicast system, which includes a dual-functional radar-communication BS, $K$ single-antenna legitimate users, and a potential eavesdropper. The BS employs an analog beamforming architecture with phase shifters driven by a single variable gain amplifier (VGA) to multicast a confidential message to all legitimate users. 
The BS is equipped with an $N_t$-element transmit uniform linear array (ULA) and an $N_r$-element receive ULA. 
Since transmit-side rotation affects both sensing illumination and secure multicast transmission, we focus on reconfiguring the transmit array and employ a conventional fixed receive ULA for echo acquisition. 
The transmit ULA is mounted on a rotatable platform, and each transmit antenna element can independently
adjust its boresight orientation.}

The origin of the global coordinate system is placed at the center of the transmit ULA. Before rotation, the transmit ULA lies along the \(x\)-axis, with its initial broadside pointing toward the positive \(y\)-axis. Accordingly, the position of the \(n\)-th transmit antenna element is given by
\begin{equation}
\bar{\mathbf c}_n=
[m^{\rm t}_nd,\;0]^T,n\in\{1,2,\ldots,N_t\},
\end{equation}
where \(d=\frac{\lambda}{2}\) denotes inter-element spacing with $\lambda$ being the wavelength, and $m^{\rm t}_n=n-\frac{N_t+1}{2}, n \in \{1,2,...,N_t\}$ is the normalized position index of the $n$-th transmit RA element.
Since a two-dimensional system is considered, the array-level rotation can be described by a single angle \(\phi_{\rm arr}\in[-\phi_{\max},\phi_{\max}]\), where \(\phi_{\max}\) denotes the maximum allowable array-level rotation angle. In addition to the global array rotation, each transmit RA element can independently adjust its boresight within the \(x\)-\(y\) plane. 
Its local boresight direction is represented by
\begin{equation}
    \mathbf f_n=[\sin\varphi_n, \cos\varphi_n]^T, \forall n\in\{1,2,\ldots,N_t\},
\end{equation}
where \(\varphi_n\) denotes the local pointing angle with respect to the local broadside direction of the rotated transmit ULA.
To reflect practical rotation limitation and avoid excessive antenna coupling \cite{Kumar2023}, the local pointing angle is constrained as follows:
\begin{equation}
-\varphi_{\max} \le \varphi_n\le \varphi_{\max},\forall n,
\end{equation}
where \(\varphi_{\max}\in[0,\pi/2]\) denotes the maximum allowable rotation angle of each RA element.
By stacking the local boresight vectors of all transmit RA elements, the local boresight matrix is defined as follows:
\begin{equation}
\mathbf F \triangleq
\left[
\mathbf f_1,\mathbf f_2,\ldots,\mathbf f_{N_t}
\right]\in\mathbb R^{2\times N_t}.
\end{equation}

After the array-level rotation, the position of the \(n\)-th transmit antenna element in the global coordinate system becomes
\begin{equation}
\mathbf c_n=\mathbf R(\phi_{\rm arr})\bar{\mathbf c}_n,n\in\{1,2,\ldots,N_t\},
\end{equation}
where
\begin{equation}
    \mathbf R(\phi_{\rm arr})
    =
    \begin{bmatrix}
        \cos\phi_{\rm arr} & \sin\phi_{\rm arr}\\
        -\sin\phi_{\rm arr} & \cos\phi_{\rm arr}
    \end{bmatrix}
    \in\mathbb R^{2\times2}
\end{equation}
is the rotation matrix. Accordingly, the boresight matrix in the
global coordinate system is given by
\begin{equation}
\bar{\mathbf F}
=
\mathbf R(\phi_{\rm arr})\mathbf F
\in\mathbb R^{2\times N_t}.
\end{equation}

Thus, the spatial configuration of each transmit RA is jointly determined by the array-level rotation \(\phi_{\rm arr}\) and the element-level pointing angle \(\varphi_n\).
\subsection{Channel Model}
We consider a far-field quasi-static line-of-sight (LoS) channel model, which is
representative of millimeter-wave (mmWave) systems where the LoS
component is often dominant due to severe propagation loss and
relatively sparse scattering at high frequencies \cite{Cao2025}.
The position vector of the \(k\)-th legitimate user is denoted by
\(\mathbf q_k=[r_k\sin\psi_k,\;r_k\cos\psi_k]^T\),
where \(r_k\) is the distance between user \(k\) and the BS, and \(\psi_k \in (-\frac{\pi}{2},\frac{\pi}{2})\) denotes the azimuth angle of user \(k\) measured from the positive $y$-axis in the global coordinate system. Accordingly, the transmit steering vector of the rotatable ULA toward direction \(\psi_k\) can be expressed as follows:
\begin{equation}
\mathbf a_t(\tilde{\psi}_{k})=\frac{1}{\sqrt{N_t}}
\left[
e^{j\pi m^{\rm t}_1 \sin(\tilde{\psi}_{k})},
\ldots,
e^{j\pi m^{\rm t}_{N_t} \sin(\tilde{\psi}_{k})}
\right]^T,
\end{equation}
where \(\tilde{\psi}_{k} = \psi_k-\phi_{\rm arr}\) represents the user direction relative to the rotated transmit array. The receive steering vector depends only on the incident direction. Therefore, the receive steering vector is given by
\begin{equation}
\mathbf a_r(\psi_k)=\frac{1}{\sqrt{N_r}}
\left[
e^{j\pi m^{\rm r}_1 \sin\psi_k},
\ldots,
e^{j\pi m^{\rm r}_{N_r} \sin\psi_k}
\right]^T,
\end{equation}
where $m^{\rm r}_n =n-\frac{N_r+1}{2}, n \in \{1,2,...,N_r\}$ is the normalized position index of the $n$-th receive antenna element.

The effective antenna gain of each RA element depends on both its boresight orientation and the signal propagation direction. In this paper, we consider the following widely used directional gain pattern for each RA element\cite{Balanis2015}:
\begin{equation} \label{G_origin}
G(\theta)=
\begin{cases}
G_0 \cos^{2p}\!(\theta), & \theta\in\left[0,\frac{\pi}{2}\right],\\
0, & \text{otherwise},
\end{cases}
\end{equation}
where \(G_0=2(2p+1)\) is the maximum boresight gain to satisfy the law of power conservation, and $p\geq 0$ is the directivity factor that characterizes the beamwidth of the antenna main lobe.
% Accordingly, the $n$-th RA's directional gain in the direction of the user $k$ is given by
% \begin{equation}
%     G_{k,n}=G_0 \cos^{2p}\!(\theta_{k,n}),
% \end{equation}
% where $\cos(\theta_{k,n}) \triangleq \bar {\mathbf f}_n^T {\mathbf u}_{k,n}$ is the projection of the direction vector of the $k$-th user ${\mathbf u}_{k,n}\triangleq \frac{\mathbf q_k-\mathbf c_n}{\| \mathbf q_k-\mathbf c_n\| }$ onto the pointing vector of the $n$-th RA. Although the propagation phase is modeled under the far-field approximation, the element-wise directional gain is evaluated using the exact geometric direction from each transmit antenna element to the user.
Accordingly, the $n$-th RA's directional gain in the direction of user $k$
is given by
\begin{equation}
    G_{k,n}=G_0\left[\bar{\mathbf f}_n^T\mathbf u_{k,n}\right]_+^{2p},
    \label{G_kn}
\end{equation}
where $[x]_+\triangleq\max\{x,0\}$, 
$\bar{\mathbf f}_n^T\mathbf u_{k,n}=\cos(\theta_{k,n})$, and
${\mathbf u}_{k,n}\triangleq
\frac{\mathbf q_k-\mathbf c_n}{\|\mathbf q_k-\mathbf c_n\|}$
is the exact geometric direction from the $n$-th transmit RA element
to user $k$. 
Consequently, the channel coefficient from the $n$-th transmit RA element to user $k$ is expressed as follows:
\begin{equation}
    h_{k,n}=\beta_k \sqrt{G_{k,n}}a_{k,n},
\end{equation}
where $\beta_k= \frac{\lambda}{4 \pi r_k}e^{j\frac{2 \pi r_k}{\lambda}}$ is the complex-valued path gain from the BS to user $k$, and $a_{k,n}$ denotes the $n$-th entry of $\mathbf a_t(\tilde \psi_k)$.
Hence, the communication channel vector $\mathbf h_k \in \mathbb{C}^{N_t \times 1} $ between user $k$ and the BS can be written as follows:
\begin{equation}
    \mathbf h_{k}=[h_{k,1},h_{k,2},...,h_{k,N_t}]^T.
\end{equation}
Similarly, the channel coefficient from the $n$-th transmit RA element to the eavesdropper can be defined in the same manner, and the corresponding channel vector is denoted by $\mathbf h_e \in \mathbb{C}^{N_t\times 1}$, where $\beta_e= \frac{\lambda}{4 \pi r_e}e^{j\frac{2 \pi r_e}{\lambda}}$ denotes the complex-valued path gain between the BS and the eavesdropper, and $r_e$ denotes the BS-eavesdropper distance.

\subsection{Signal Model}
Let $s \in \mathbb C$ be the coded confidential information symbol for users with $\mathbb{E} [|s|^2]=1$. The transmit signal is given by
\begin{equation}
    \mathbf x = \sqrt{P_t} \mathbf w s \in \mathbb C^{N_t},
\end{equation}
where $P_t>0$ represents the communication transmit power and $\mathbf w\in \mathbb C^{N_t}$ denotes the analog beamforming vector implemented by the phase shifters with constant modulus (CM), i.e.,
$
    |w_n| = \frac{1}{\sqrt{N_t}}, \forall n=1,...,N_t.
$

Then the received signals at the $k$-th user and the eavesdropper are respectively given by
\begin{align}
{y}_k &= \sqrt{P_t}\,\mathbf{h}_k^{H}\mathbf{w}s + n_k,\\
{y}_e &= \sqrt{P_t}\,\mathbf{h}_e^{H}\mathbf{w}s + n_e,
\end{align}
where $n_k\sim\mathcal{CN}(0,\sigma_c^2)$ and
$n_e\sim\mathcal{CN}(0,\sigma_e^2)$ denote independent zero-mean
additive white Gaussian noise (AWGN) samples with variances $\sigma_c^2$ and $\sigma_e^2$,
respectively.

Hence, the achievable rate of user $k$ is given by
\begin{equation}
    R_k=\log_2 (1+ \gamma_c  |\mathbf h_k^H \mathbf w |^2),
\end{equation}
where $\gamma_c = P_t/\sigma_c^2$ denotes the transmit signal-to-noise ratio (SNR).
Similarly, the achievable rate of the eavesdropper is 
\begin{equation}
    R_e=\log_2 (1+ \gamma_e  |\mathbf h_e^H \mathbf w |^2).
\end{equation}

According to the information-theoretic principles of secure communication \cite{L2017}, the achievable secrecy rate for user $k$ can be written as follows:
\begin{equation}
    R_k^{sec}= [R_k-R_e]^+,
\end{equation}

\section{Sensing CRB Analysis and Problem Formulation}
In this section, we first establish a theoretical benchmark for the sensing performance of the BS by deriving the CRB for the estimation of the eavesdropper’s AoD.
Building on this analysis, we then formulate a robust design problem that aims to maximize the worst-case secrecy rate under the derived estimation uncertainty by jointly optimizing the RA rotations and the analog beamforming vector.
\subsection{Echo Signal Model and Angle Estimation}
In the considered system, cooperative users periodically
report their locations to the BS via GPS-based feedback
\cite{Yang2025}, whereas the direction of the passive eavesdropper is unavailable. Therefore the BS estimates the
eavesdropper's AoD from the received echo signals before secure
transmission.
For a given RA configuration $\Phi_s =
\left(
\phi_{\mathrm{arr}}^s,
\bm \varphi^s
\right)$, the effective transmit sensing response can be written as follows:
\begin{equation}
\mathbf b(\theta)
\triangleq
\operatorname{diag}\!\left(
\sqrt{G_1(\theta)},\ldots,
\sqrt{G_{N_t}(\theta)}
\right)
\mathbf a_t\!\left(\theta-\phi_{\mathrm{arr}}^s\right).
\label{sensing_response}
\end{equation}
At the $l$-th snapshot $(1\leq l\leq L)$, the received echo signal is
\begin{equation}
\mathbf y_s(l)=\sqrt{P_s}\,\beta_s\,\mathbf a_r(\theta_e)\mathbf b^H(\theta_e) \mathbf x_s(l)+\mathbf n_s(l),
\end{equation}
where \(\theta_e\) is the true angular direction of the eavesdropper, \(P_s\) denotes the sensing transmit power, \(\mathbf x_s(l)\) is the probing symbol satisfying \(\mathbb E[\|\mathbf x_s(l)\|^2]=1\),
$\beta_s=
\sqrt{\frac{\lambda^2\alpha}{64\pi^3r_e^4}}
\exp\!\left(j\frac{4\pi r_e}{\lambda}\right)$
denotes the round-trip sensing channel coefficient, \(\alpha\) is the radar cross section (RCS) of the target, \(r_e\) is the distance between the BS and the eavesdropper, and \(\mathbf n_s(l)\sim\mathcal{CN}(\mathbf 0,\sigma_s^2\mathbf I_{N_r})\) is the additive white Gaussian noise.

By collecting all the received echoes over the \(L\) sensing beams, the sensing observations can be stacked into the following matrix form:
\begin{equation}
\begin{split}
\mathbf Y_s
&=
\big[\mathbf y_s(1),\mathbf y_s(2),\ldots,\mathbf y_s(L)\big] \\
&=
\beta_s\mathbf a_r(\theta_e)\mathbf b^H(\theta_e)\mathbf X_s+\mathbf N_s,
\end{split}
\end{equation}
where $\mathbf X_s=
\sqrt{P_s}\,[\mathbf x_s(1),\mathbf x_s(2),\ldots,\mathbf x_s(L)] \in \mathbb C^{N_t \times L}$,
and $\mathbf N_s=
[\mathbf n_s(1),\mathbf n_s(2),\ldots,\mathbf n_s(L)]$.
The transmit covariance matrix associated with \(\mathbf X_s\) can be expressed as follows:
\begin{equation}\label{Rx}
\mathbf R_X=\frac{1}{L}\mathbf X_s\mathbf X_s^H=\frac{P_s}{N_t}\mathbf I_{N_t},
\end{equation}
Note that the equation in \eqref{Rx} holds only when $L\geq N_t$ \cite{Lu2024,Cao2025}.
Based on \(\mathbf Y_s\), the BS employs MLE to estimate the direction of the single eavesdropper owing to its favorable statistical efficiency and its asymptotic ability to approach the CRB at sufficiently high sensing SNR.
\begin{lemma}
The MLE of the eavesdropper direction is given by
\begin{equation} \label{MLE}
\hat{\theta}_e
=
\arg\max_{\theta\in\Theta_1}
\frac{\left|
\mathbf a_r^H(\theta)\mathbf Y_s\mathbf X_s^H\mathbf b(\theta)
\right|^2}{\mathbf b^H(\theta) \mathbf X_s \mathbf X^H_s \mathbf b(\theta)}
,
\end{equation}
where $\Theta_1 =[\theta_{\mathrm L},\theta_{\mathrm U}]$ denotes the angular search region. It is assumed
to be coarsely obtained by an external tracker before the sensing
stage.

\end{lemma}

\begin{proof}
See Appendix A.
\end{proof}

The above lemma provides a point estimate of the eavesdropper direction. However, due to noise and finite sensing resources, the estimate \(\hat{\theta}_e\) is generally imperfect. 
\subsection{CRB-Based Angular Uncertainty Modeling}
To quantify the accuracy of the angle estimate $\hat{\theta}_e$, we
consider the mean-square error (MSE)
$\mathbb{E}\{(\theta_e-\hat{\theta}_e)^2\}$. Since its closed-form
expression is difficult to obtain, we employ the CRB as a tractable
accuracy metric, which serves as an asymptotically tight lower bound
on the MSE at moderate-to-high sensing SNRs.

\begin{lemma}\label{lemma2}
For a sensing configuration $\Phi_s$, the CRB for estimating the eavesdropper direction is
\begin{equation}
\operatorname{CRB}_{\theta}
\left(
\theta;
\Phi_s
\right)
=
\frac{N_t\sigma_s^2}
{
\displaystyle
2LP_s|\beta_s|^2
\Psi
\left(
\theta;
\Phi_s
\right)
},
\label{CRB}
\end{equation}
where the effective angular information is
\begin{equation}
\begin{aligned}
&\Psi
\left(
\theta;
\Phi_s
\right)
={}
\left\|
\dot{\mathbf a}_r(\theta)
\right\|^2
\left\|
\mathbf b
\left(
\theta;
\Phi_s
\right)
\right\|^2
+
\left\|
\dot{\mathbf b}
\left(
\theta;
\Phi_s
\right)
\right\|^2\\
&+
2\operatorname{Re}
\Bigg\{
\mathbf a_r^H(\theta)
\dot{\mathbf a}_r(\theta)
\,
\mathbf b^H
\left(
\theta;
\Phi_s
\right)
\dot{\mathbf b}
\left(
\theta;
\Phi_s
\right)
\Bigg\}\\
&-
\frac{|
\mathbf a_r^H(\theta)
\dot{\mathbf a}_r(\theta)
\|
\mathbf b
(
\theta;
\Phi_s
)
\|^2
+
\dot{\mathbf b}^{H}
(
\theta;
\Phi_s
)
\mathbf b
(
\theta;
\Phi_s
)
|^2
}
{
\left\|
\mathbf b
\left(
\theta;
\Phi_s
\right)
\right\|^2
},
\label{effective_angular_information}
\end{aligned}
\end{equation}
with
$
\dot{\mathbf a}_r(\theta)
\triangleq
\frac{\partial\mathbf a_r(\theta)}{\partial\theta}$ and
$\dot{\mathbf b}
\left(
\theta;
\Phi_s
\right)
\triangleq
\frac{
\partial
\mathbf b
\left(
\theta;
\Phi_s
\right)
}{
\partial\theta
}$.
\end{lemma}

\begin{proof}
See Appendix~B.
\end{proof}

Lemma~\ref{lemma2} shows that the sensing-stage RA rotations
affect the CRB through the transmit response
$\mathbf b$ and its angular derivative $\dot{\mathbf b}$.
The array rotation $\phi_{\mathrm{arr}}^s$ changes both the
transmit phase progression and the directional gains of all
RA elements. In contrast, each local rotation
$\varphi_n^s$ only adjusts the directional gain of the
corresponding element. Therefore, $\Psi$ jointly captures the
transmit response power, its angular sensitivity, and their
correlation. This coupling motivates the joint optimization of
the array-level and element-level rotations to minimize the CRB.

The eavesdropper direction $\theta_e$ is treated as an unknown
deterministic parameter, while $\Theta_1$ specifies only its search
region rather than a prior distribution. To
facilitate a tractable robust design, we approximate
the MLE error as a zero-mean Gaussian random variable with variance
given by the corresponding CRB \cite{Su2023}.
Specifically, the AoD estimation error is modeled as
$    \theta_e-\hat{\theta}_e \sim 
    \mathcal N\!\left(0,\mathrm{CRB}(\hat{\theta}_e)\right).
$
Under this model, the true eavesdropper direction lies in the uncertainty region
$\Theta_2 $ with probability approximately $0.9973$\cite{Chandola2009}, where $\Theta_2 =[\hat{\theta}_e-3\sqrt{\mathrm{CRB}(\hat{\theta}_e)},\hat{\theta}_e+3\sqrt{\mathrm{CRB}(\hat{\theta}_e)}]$.
Since the CRBs are determined
by the RA configuration, this AoD estimation range is also a
function of the antenna geometry. This uncertainty set specifies the
eavesdropper directions considered by the subsequent worst-case
secure communication design. The sensing-stage RA rotations therefore
affect secure transmission through the uncertainty passed to
the communication stage.

In principle, both the distance and the direction of the eavesdropper
can be inferred from the sensing observations through the echo delay
and spatial response, respectively. In the considered far-field
model, distance mainly affects the path gain, whereas angular
uncertainty changes the array response and may cause beamforming
mismatch. Accordingly, as in \cite{Cao2025}, we neglect distance
estimation errors and focus on angular estimation errors in the
subsequent robust design.
\subsection{Problem Formulation}
\label{subsec:problem_formulation}
Leveraging the CRB analysis, our goal is to maximize the worst-case secrecy rate by jointly designing the sensing and communication stages.
In the sensing stage, the RA configuration $\Phi_s$ is optimized to minimize the CRB and narrow the uncertainty region $\Theta_2$. A smaller uncertainty region allows for more precisely targeted robust beamforming \cite{Lyu2024}.
In the communication stage, given the estimated direction $\hat{\theta}_e$ and the uncertainty region $\Theta_2$, the RA configuration $\Phi_c$ and the analog beamforming vector $\mathbf w$ are jointly optimized to maximize the worst-case secrecy rate over all possible eavesdropper directions.
The overall joint optimization problem can be formulated as
\begin{subequations}
\label{overall_problem}
\begin{align}
\max_{\Phi_s,\Phi_c,\mathbf w}\quad
&
\min_{\substack{
k\in\mathcal K,\,
\hat{\theta}_e\in\Theta_1,\,
\theta_e\in\Theta_2
}}
\left[
R_k(\mathbf w,\Phi_c)
-
R_e(\theta_e;\mathbf w,\Phi_c)
\right]^+
\label{overall_problem_objective}
\\
\mathrm{s.t.}\quad
&
|w_n|=\frac{1}{\sqrt{N_t}},
\qquad \forall n,
\label{overall_problem_CM}
\\
&
|\phi_{\mathrm{arr}}^q|
\leq
\phi_{\max},
\qquad q\in\{s,c\},
\label{overall_problem_array}
\\
&
|\varphi_n^q|
\leq
\varphi_{\max},
\qquad
\forall n,\quad q\in\{s,c\}.
\label{overall_problem_element}
\end{align}
\end{subequations}
Here, constraint \eqref{overall_problem_CM} enforces the CM constraint, while \eqref{overall_problem_array} and \eqref{overall_problem_element} specify the feasible array-level and element-level rotations, respectively.

Problem \eqref{overall_problem} is challenging due to the dependence
of $\Theta_2$ on the sensing configuration and the MLE output, as
well as the nested worst-case objective and non-convex constraints.
To address this challenge, we decompose the problem into two sequential subproblems.

First, the sensing-stage RA configuration is optimized to minimize
the CRB, which implicitly improves the worst-case secrecy rate by reducing estimation uncertainty. According to \eqref{CRB}, this is
equivalent to maximizing the minimum effective angular information,
which yields
\begin{subequations}
\begin{align}
\max_{\Phi_s}\quad
&
\min_{\theta\in\Theta_1}
\Psi(\theta;\Phi_s)
\\
\mathrm{s.t.}\quad
&
|\phi_{\mathrm{arr}}^s|
\leq
\phi_{\max},
\\
&
|\varphi_n^s|
\leq
\varphi_{\max},
\qquad \forall n.
\end{align}
\label{sensing_problem}
\end{subequations}

After obtaining $\Phi_s^\star$, the BS estimates
$\hat{\theta}_e$ according to \eqref{MLE} and constructs the
corresponding uncertainty region $\Theta_2$. The communication-stage
RA configuration and analog beamformer are then optimized by solving the following problem:
\begin{subequations}
\label{communication_problem}
\begin{align}
\max_{\Phi_c,\mathbf w}\quad
&
\min_{\substack{
k\in\mathcal K,
\theta_e\in\Theta_2}}
\left[
R_k(\mathbf w,\Phi_c)
-
R_e(\theta_e;\mathbf w,\Phi_c)
\right]^+
\label{communication_problem_objective}
\\
\mathrm{s.t.}\quad
&
|\phi_{\mathrm{arr}}^c|
\leq
\phi_{\max},
\label{communication_problem_array}
\\
&
|\varphi_n^c|
\leq
\varphi_{\max},
\qquad \forall n,
\label{communication_problem_element}\\
& \eqref{overall_problem_CM}. \notag
\end{align}
\end{subequations}

The solutions to these two subproblems are detailed in the following sections. 
This scheme further demonstrates the functional versatility of RAs by adaptively adjusting the array and element rotation angles to prioritize sensing accuracy in the first stage and communication security in the second stage.

\section{Optimization of Sensing-Stage RA Configuration}
\label{sec:sensing_optimization}

Problem~\eqref{sensing_problem} is a non-convex semi-infinite
program (SIP), since a finite number of RA rotation variables are
optimized against infinitely many directions in the continuous
region $\Theta_1$. Moreover, the minimum operator renders the
objective non-smooth. To address this issue, we
introduce an auxiliary variable $t$, which represents a lower
bound on the effective angular information over $\Theta_1$.
Problem~\eqref{sensing_problem} is then equivalently reformulated as follows:
\begin{subequations}
\begin{align}
\max_{\Phi_s,t}\quad
& t
\\
\mathrm{s.t.}\quad
& t\leq\Psi(\theta;\Phi_s),
\qquad \forall\theta\in\Theta_1,
\label{sensing_hypograph_information}
\\
& \left|\phi_{\mathrm{arr}}^s\right|
\leq\phi_{\max},
\label{sensing_hypograph_array}
\\
& \left|\varphi_n^s\right|
\leq\varphi_{\max},
\qquad \forall n.
\label{sensing_hypograph_element}
\end{align}
\label{sensing_hypograph_problem}
\end{subequations}

To handle the continuously indexed constraints, we
discretize $\Theta_1$ into the angular set
\begin{equation}
\mathcal T
=
\left\{
\theta_{m_1}
=
\theta_{\mathrm L}
+
m_1\Delta_{\theta_1}
\ \middle|\
m_1=0,\ldots,M_{1}
\right\},
\label{sensing_angular_traversal_set}
\end{equation}
where
$\Delta_{\theta_1}
=
\frac{\theta_{\mathrm U}-\theta_{\mathrm L}}{M_{1}}$. By replacing the continuously indexed constraints in
\eqref{sensing_hypograph_problem} with those evaluated over
$\mathcal T$, we obtain the finite-dimensional approximation as follows:
\begin{subequations}
\label{sensing_discretized_problem}
\begin{align}
\max_{\Phi_s,t}\quad
&
t
\\
\mathrm{s.t.}\quad
&
t\leq\Psi(\theta_{m_1};\Phi_s),
\qquad
\forall\theta_{m_1}\in\mathcal T,
\label{sensing_discretized_information}
\\
&
\left|\phi_{\mathrm{arr}}^s\right|
\leq\phi_{\max},
\\
&
\left|\varphi_n^s\right|
\leq\varphi_{\max},
\qquad \forall n.
\end{align}
\end{subequations}

For fixed $\Phi_s$, the optimal auxiliary variable in
\eqref{sensing_discretized_problem} is
\begin{equation}
t
=
\mathcal J(\Phi_s)
\triangleq
\min_{\theta_{m_1}\in\mathcal T}
\Psi(\theta_{m_1};\Phi_s).
\label{sensing_discretized_objective_function}
\end{equation}
Hence, problem~\eqref{sensing_discretized_problem} is equivalent
to maximizing $\mathcal J(\Phi_s)$ over the RA rotation
variables. Since the array-level and element-level rotations are
nonlinearly coupled through $\mathbf b(\theta_{m_1};\Phi_s)$ and
$\dot{\mathbf b}(\theta_{m_1};\Phi_s)$, we employ a cyclic coordinate
search method. At each coordinate update, the method exhaustively
evaluates the candidate values of one rotation variable while
holding all remaining rotation variables fixed.

Define the candidate sets for the array-level and element-level
rotation angles as follows:
\begin{align}
\mathcal G_{\mathrm{arr}}
=
\left\{
-\phi_{\max}
+
\ell\Delta_{\mathrm{arr}}
\ \middle|\
\ell=0,\ldots,L_{\mathrm{arr}}
\right\},
\label{sensing_array_candidate_set}\\
\mathcal G_{\mathrm{ele}}
=
\left\{
-\varphi_{\max}
+
\ell\Delta_{\mathrm{ele}}
\ \middle|\
\ell=0,\ldots,L_{\mathrm{ele}}
\right\},
\label{sensing_element_candidate_set}
\end{align}
respectively, where
$
\Delta_{\mathrm{arr}}
=
\frac{2\phi_{\max}}{L_{\mathrm{arr}}}
$, $
\Delta_{\mathrm{ele}}
=
\frac{2\varphi_{\max}}{L_{\mathrm{ele}}}.
$

Let $\Phi_s^{(r)}$
denote the sensing-stage RA configuration at the $r$-th coordinate
search cycle.
With $\boldsymbol{\varphi}^{s,(r)}$ fixed, the array-level rotation
is updated as follows:
\begin{equation}
\phi_{\mathrm{arr}}^{s,(r+1)}
\in
\arg\max_{\phi\in\mathcal G_{\mathrm{arr}}}
\mathcal J
\left(
\phi,\boldsymbol{\varphi}^{s,(r)}
\right).
\label{sensing_array_coordinate_update}
\end{equation}
For every candidate $\phi$, the objective in
\eqref{sensing_array_coordinate_update} is evaluated as follows:
\begin{equation}
\mathcal J
\left(
\phi,\boldsymbol{\varphi}^{s,(r)}
\right)
=
\min_{\theta_{m_1}\in\mathcal T}
\Psi
\left(
\theta_{m_1};
\phi,\boldsymbol{\varphi}^{s,(r)}
\right).
\end{equation}
After updating the array-level rotation, the element-level
rotations are optimized sequentially. To update the $n$-th
element, define the candidate element-rotation vector as follows:
\begin{equation}
\boldsymbol{\varphi}_{n}^{s,(r)}(\xi)
=[
\varphi_1^{s,(r+1)},\ldots,
\varphi_{n-1}^{s,(r+1)},\xi,
\varphi_{n+1}^{s,(r)},\ldots,
\varphi_{N_t}^{s,(r)}
]^T.
\label{sensing_partial_element_vector}
\end{equation}
The rotation angle of the $n$-th element is updated as follows:
\begin{equation}
\varphi_n^{s,(r+1)}
\in
\arg\max_{\xi\in\mathcal G_{\mathrm{ele}}}
\mathcal J
\left(
\phi_{\mathrm{arr}}^{s,(r+1)},
\boldsymbol{\varphi}_{n}^{s,(r)}(\xi)
\right).
\label{sensing_element_coordinate_update}
\end{equation}
Each candidate element rotation is evaluated according to
its worst sensing performance over the entire discretized angular
set $\mathcal T$. Since the current value of each rotation
variable is included in its candidate set, every coordinate
update preserves or improves the objective value. It follows that
\begin{equation}
\mathcal J\left(\Phi_s^{(r+1)}\right)
\geq
\mathcal J\left(\Phi_s^{(r)}\right).
\label{sensing_coordinate_monotonicity}
\end{equation}
The coordinate-search iterations terminate when
\begin{equation}
\frac{
\mathcal J\left(\Phi_s^{(r+1)}\right)
-
\mathcal J\left(\Phi_s^{(r)}\right)
}{
\max\left\{
1,
\left|
\mathcal J\left(\Phi_s^{(r)}\right)
\right|
\right\}
}
\leq
\epsilon_{\mathrm c},
\label{sensing_coordinate_stopping}
\end{equation}
where $\epsilon_{\mathrm c}>0$ denotes the convergence tolerance.
After convergence, the optimized RA configuration
and its corresponding worst-case effective angular information
over $\mathcal T$ are given by
\begin{equation}
\Phi_s^\star
=
\Phi_s^{(r+1)},
t^\star
=
\mathcal J\left(\Phi_s^\star\right).
\end{equation}

The overall sensing-stage optimization procedure is summarized in
Algorithm~\ref{alg:sensing_coordinate_search}.
\begin{algorithm}[t]
\caption{Cyclic Coordinate Search for Sensing-Stage RA Design}
\label{alg:sensing_coordinate_search}
\begin{algorithmic}[1]
\State \textbf{Input:}
$\Theta_1$,
$M_1$,
$L_{\mathrm{arr}}$,
$L_{\mathrm{ele}}$,
$r=0$,
and $\epsilon_{\mathrm c}$.
\State Construct
$\mathcal T$,
$\mathcal G_{\mathrm{arr}}$, and
$\mathcal G_{\mathrm{ele}}$
via
\eqref{sensing_angular_traversal_set},
\eqref{sensing_array_candidate_set}, and
\eqref{sensing_element_candidate_set}.
\State Select an initial configuration
$\Phi_s^{(0)}
\in
\mathcal G_{\mathrm{arr}}
\times
\mathcal G_{\mathrm{ele}}^{N_t}$.
\Repeat
    \State Compute
    $J_{\mathrm{old}}
    \leftarrow
    \mathcal J(\Phi_s^{(r)})$.
    \State Update
    $\phi_{\mathrm{arr}}^{s,(r+1)}$
    according to
    \eqref{sensing_array_coordinate_update}.
    \For{$n=1,\ldots,N_t$}
        \State Update
        $\varphi_n^{s,(r+1)}$
        according to
        \eqref{sensing_element_coordinate_update}.
    \EndFor
    \State Set
    $\Phi_s^{(r+1)}
    \leftarrow
    \left(
    \phi_{\mathrm{arr}}^{s,(r+1)},
    \boldsymbol{\varphi}^{s,(r+1)}
    \right)$.
    \State Compute
    $J_{\mathrm{new}}
    \leftarrow
    \mathcal J(\Phi_s^{(r+1)})$.
    \State Set $r\leftarrow r+1$.
\Until{
    $(J_{\mathrm{new}}-J_{\mathrm{old}})
    /\max\{1,|J_{\mathrm{old}}|\}
    \leq\epsilon_{\mathrm c}$
}
\State \textbf{Output:}
$\Phi_s^\star=\Phi_s^{(r)}$,
$t^\star=\mathcal J(\Phi_s^\star)$.
\end{algorithmic}
\end{algorithm}
The angular and rotation resolutions
determine the approximation accuracy of the obtained solution.

\section{Optimization of Communication-Stage RA Configuration and Beamforming}
Since the angular uncertainty region $\Theta_2$ determines the set
of possible eavesdropper channels, jointly designing the analog
beamformer $\mathbf w$ and the communication-stage RA configuration
$\Phi_c$ provides additional spatial DoFs to improve the secrecy rate over $\Theta_2$. 
In this section, we consider the robust design for the secure communication stage, aiming to maximize the worst-case secrecy rate given the uncertainty region established in the sensing stage. 
The goal is to solve
problem~\eqref{communication_problem} by jointly optimizing
$\mathbf w$, $\phi_{\mathrm{arr}}^c$, and
$\boldsymbol{\varphi}^c$. Unlike conventional
robust beamforming with an uncertainty region specified
\emph{a priori}, $\Theta_2$ is constructed from the MLE
and the corresponding CRB under the optimized
sensing-stage RA configuration $\Phi_s^\star$. Therefore, the
communication-stage robust design explicitly adapts to the sensing
accuracy achieved.

\subsection{Problem Reformulation}
\label{subsec:communication_reformulation}
To facilitate the subsequent optimization, given the estimate
$\hat{\theta}_e$ and the corresponding uncertainty region $\Theta_2$,
we first rewrite problem~\eqref{communication_problem} as follows:
\begin{subequations}
\label{communication_untruncated_problem}
\begin{align}
\max_{\mathbf w,\Phi_c}\quad
&
\min_{k\in\mathcal K}
R_k(\mathbf w,\Phi_c)
-
\max_{\theta\in\Theta_2}
R_e(\theta;\mathbf w,\Phi_c).
\label{communication_untruncated_objective}
\\
\mathrm{s.t.}\quad
&
\eqref{overall_problem_CM},
\eqref{communication_problem_array},
\eqref{communication_problem_element}.
\notag
\end{align}
\end{subequations}
Note that the operator $[\cdot]^+$ can be omitted without loss of generality, as a feasible solution always ensures a non-negative secrecy rate \cite{Xiong2026}.
The maximization over the continuous region $\Theta_2$ makes
\eqref{communication_untruncated_problem} non-smooth when the worst-case direction changes. To obtain a
finite-dimensional formulation, let
$\Theta_2=[\xi_{\mathrm L},\xi_{\mathrm U}]$ and discretize it into $M$ uniformly sampled directions:
\begin{equation}
\Theta
=
\left\{
\theta_{m}
=
\xi_{\mathrm L}
+
(m-1)\Delta_{\theta}
\ \middle|\
{m}=1,\ldots,M
\right\},
\end{equation}
where $\Delta_{\theta} =\frac{\xi_{\mathrm U}-\xi_{\mathrm L}}{M-1}$.
For each sampled direction $\theta_{m}$, define the corresponding
eavesdropper rate as follows:
\begin{equation}
R_{e,m}(\mathbf w,\Phi_c)
\triangleq
R_e(\theta_m;\mathbf w,\Phi_c),
\qquad m=1,\ldots,M.
\label{communication_indexed_eavesdropper_rate}
\end{equation}

Accordingly, problem~\eqref{communication_untruncated_problem} can be transformed to the following:
\begin{subequations}
\label{communication_discretized_problem}
\begin{align}
\max_{\mathbf w,\Phi_c}\quad
&
\min_{\substack{
k=1,\ldots,K\\
m=1,\ldots,M
}}
\tilde R_{k,m}(\mathbf w,\Phi_c)
\label{communication_discretized_objective}
\\
\mathrm{s.t.}\quad
&
\eqref{overall_problem_CM},
\eqref{communication_problem_array},
\eqref{communication_problem_element},
\notag
\end{align}
\end{subequations}
where $
\tilde R_{k,m}(\mathbf w,\Phi_c)
\triangleq
R_k(\mathbf w,\Phi_c)
-
R_{e,m}(\mathbf w,\Phi_c)
$.
Although problem \eqref{communication_discretized_problem} is more tractable than \eqref{communication_problem}, it remains non-smooth due to the \(\min\) operator over the users and sampled directions. We therefore further approximate its objective by a smooth function based on the following lemma \cite{Boyd2004}.

\begin{lemma}\label{lemma_softmin}
For any \(c_1,c_2,\ldots,c_K\in\mathbb{R}\) and any \(\beta>0\), it holds that
\begin{equation}\label{lemma_softmin_eq}
\min_{k=1,\ldots,K} c_k-\frac{1}{\beta}\log K
\le
- \frac{1}{\beta}\log \sum_{k=1}^{K}\exp\!\left(-\beta c_k\right)
\le
\min_{k=1,\ldots,K} c_k.
\end{equation}
Moreover, the above bounds become tight as \(\beta\to \infty \).
\end{lemma}

\begin{proof}
See Appendix C.
\end{proof}

Using Lemma~\ref{lemma_softmin}, problem \eqref{communication_discretized_problem} can be smoothly approximated as follows:
\begin{subequations}
\label{P3}
\begin{align}
\max_{\mathbf w,\Phi_c}\quad
&
-\frac{1}{\beta_{c}}
\log
\sum_{k=1}^{K}
\sum_{m=1}^{M}
\exp\!\left(
-\beta_{c}\tilde R_{k,m}(\mathbf w,\Phi_c)
\right)
\label{P3_a}
\\
\mathrm{s.t.}\quad
&
\eqref{overall_problem_CM},
\eqref{communication_problem_array},
\eqref{communication_problem_element}.
\notag
\end{align}
\end{subequations}
where \(\beta_{c} > 0\) is the smoothing parameter. The objective function in \eqref{P3_a} is denoted by \(\mathcal G(\mathbf w,\Phi_c)\) for simplicity.
It follows from Lemma~\ref{lemma_softmin} that
\(\mathcal G\) is a smooth lower approximation of
\eqref{communication_discretized_objective}, with an approximation
gap bounded by \(\log(KM)/\beta_c\). This approximation becomes
exact as \(\beta_c\to\infty\).

Although problem~\eqref{P3} has a smooth objective, it remains
challenging due to the non-convex CM constraint
\eqref{overall_problem_CM} and the bounded-rotation constraints
\eqref{communication_problem_array} and
\eqref{communication_problem_element}. We observe that the CM
constraint can be naturally characterized by a complex circle
manifold, while the bounded-rotation constraints can be handled
through a projection method.
Based on this structure, we develop a product-space joint optimization (PSJO) framework to solve the reformulated
problem~\eqref{P3}.
\subsection{The Proposed PSJO Framework}
In this subsection, we develop the PSJO framework for solving
problem~\eqref{P3}. Specifically, the normalized analog beamformer
and RA rotation variables are represented over a product space
comprising a complex circle manifold and a box-constrained set, respectively.
We then derive an efficient conjugate gradient ascent (CGA) algorithm over this product space to jointly update both variable blocks.

\textit{1) Analog Beamforming Gradient:}
We first derive the analog-beamforming component of the joint
ascent direction. By introducing the auxiliary variable
\(\mathbf v=\sqrt{N_t}\mathbf w\), the CM constraint
\eqref{overall_problem_CM} is equivalently rewritten as follows:
\begin{equation}
|v_n|=1,\qquad \forall n=1,\ldots,N_t.
\label{normalized_CM_constraint}
\end{equation}
Accordingly, the feasible set of \(\mathbf v\) can be characterized
by the following complex circle manifold:
\begin{equation}
\mathcal M
=
\left\{
\mathbf v\in\mathbb C^{N_t\times 1}
\;\middle|\;
|v_n|=1,\ \forall n
\right\}.
\label{communication_complex_circle_manifold}
\end{equation}

The tangent space
of \(\mathcal M\) at \(\tilde{\mathbf v}\in\mathcal M\) is
given by
\begin{equation}
\mathcal T_{\tilde{\mathbf v}}\mathcal M
=
\left\{
\boldsymbol\xi_v\in\mathbb C^{N_t\times 1}
\;\middle|\;
\operatorname{Re}
\left\{
\boldsymbol\xi_v\odot\tilde{\mathbf v}^{*}
\right\}
=
\mathbf 0
\right\},
\label{communication_tangent_space}
\end{equation}
where \(\boldsymbol\xi_v\) denotes a tangent vector at
\(\tilde{\mathbf v}\).

With a slight abuse of notation, we use
\(\mathcal G(\mathbf v,\Phi_c)\) to represent
\(\mathcal G(\mathbf v/\sqrt{N_t},\Phi_c)\).
To obtain the Riemannian gradient, we first derive the Euclidean
gradient of \(\mathcal G(\mathbf v,\Phi_c)\). 
Define
\begin{equation}
a_k
=
\frac{\exp(-\beta_cR_k)}
{\sum_{i=1}^{K}\exp(-\beta_cR_i)},
b_m
=
\frac{\exp(\beta_cR_{e,m})}
{\sum_{q=1}^{M}\exp(\beta_cR_{e,q})}.
\label{communication_smoothing_weights}
\end{equation}
The Euclidean gradient of the objective function $\nabla_{\mathbf v}\mathcal G$ can be given by \eqref{grad_w}.
\begin{figure*}[!b]
\hrulefill
\begin{equation}
\begin{aligned}
\nabla_{\mathbf v}\mathcal G
={}&
\frac{2\gamma_c}{N_t\ln 2}
\sum_{k=1}^{K}
a_k
\frac{
\mathbf h_k\mathbf h_k^H\mathbf v
}{
1+\frac{\gamma_c}{N_t}
\left|\mathbf h_k^H\mathbf v\right|^2
}
-
\frac{2\gamma_e}{N_t\ln 2}
\sum_{m=1}^{M}
b_m
\frac{
\mathbf h_{e,m}\mathbf h_{e,m}^H\mathbf v
}{
1+\frac{\gamma_e}{N_t}
\left|\mathbf h_{e,m}^H\mathbf v\right|^2
}.
\end{aligned}
\label{grad_w}
\end{equation}
\end{figure*}
The Riemannian gradient is then obtained by projecting
\(\nabla_{\mathbf v}\mathcal G\) onto
\(\mathcal T_{\tilde{\mathbf v}}\mathcal M\), yielding
\begin{equation}
\begin{aligned}
\operatorname{grad}_{\mathbf v}
&\mathcal G(\tilde{\mathbf v},\Phi_c)
=\operatorname{Proj}_{\tilde{\mathbf v}}
(
\nabla_{\mathbf v}
\mathcal G(\tilde{\mathbf v},\Phi_c)
)
\\
={}&
\nabla_{\mathbf v}
\mathcal G(\tilde{\mathbf v},\Phi_c)
-
\operatorname{Re}
\{
\nabla_{\mathbf v}
\mathcal G(\tilde{\mathbf v},\Phi_c)
\odot\tilde{\mathbf v}^{*}
\}
\odot\tilde{\mathbf v}.
\end{aligned}
\label{Riemannian_grad_w}
\end{equation}

\textit{2) RA Rotation Gradient and Box Projection:}
We next derive the RA-rotation component of the joint ascent
direction. For notational convenience, define
\begin{equation}
\varsigma_k
\triangleq
\mathbf h_k^H\mathbf v,
\qquad
\varsigma_{e,m}
\triangleq
\mathbf h_{e,m}^H\mathbf v,
\label{communication_channel_response}
\end{equation}
and
\begin{equation}
\dot{\mathbf f}_n
\triangleq
\frac{\partial\mathbf f_n}
{\partial\varphi_n^c}
=
\begin{bmatrix}
\cos\varphi_n^c&
-\sin\varphi_n^c
\end{bmatrix}^{T}.
\label{communication_boresight_derivative}
\end{equation}
To characterize the dependence of the directional gain on the
array-level rotation, define
\begin{equation}
\begin{aligned}
\zeta_{q,n}
&\triangleq
\frac{
\partial
\bigl(
\bar{\mathbf f}_n^T\mathbf u_{q,n}
\bigr)
}{
\partial\phi_{\mathrm{arr}}^c
}
=
\mathbf f_n^T
\dot{\mathbf R}^{T}
\bigl(\phi_{\mathrm{arr}}^c\bigr)
\mathbf u_{q,n}
\\
&\quad-
\frac{1}{r_{q,n}}
\bar{\mathbf f}_n^T
\bigl(
\mathbf I_2-\mathbf u_{q,n}\mathbf u_{q,n}^T
\bigr)
\dot{\mathbf R}
\bigl(\phi_{\mathrm{arr}}^c\bigr)
\bar{\mathbf c}_n ,
\end{aligned}
\label{communication_zeta_definition}
\end{equation}
where
\(\dot{\mathbf R}(\phi_{\mathrm{arr}}^c)
\triangleq
\partial\mathbf R(\phi_{\mathrm{arr}}^c)/
\partial\phi_{\mathrm{arr}}^c\).
Collecting the array-level and element-level components, the
gradient of $\mathcal G$ with respect to the RA configuration is
expressed as follows:
\begin{equation}
\nabla_{\Phi_c}\mathcal G
={}
\begin{bmatrix}
\displaystyle
\frac{\partial\mathcal G}
{\partial\phi_{\mathrm{arr}}^c}
&
\displaystyle
\frac{\partial\mathcal G}
{\partial\varphi_1^c}
&
\cdots
&
\displaystyle
\frac{\partial\mathcal G}
{\partial\varphi_{N_t}^c}
\end{bmatrix}^{T}.\label{grad_RA_rotation}
\end{equation}
The required channel-response derivatives and the corresponding
objective gradients are given in
\eqref{channel_array_rotation_derivative}--\eqref{grad_element_rotation}.
\begin{figure*}[!b]
\hrulefill
\begin{align}
\frac{\partial[\mathbf h_q]_n}
{\partial\phi_{\mathrm{arr}}^c}
={}&
\beta_q\sqrt{G_0}
\left[
\mathbf a_t(\tilde{\psi}_q)
\right]_n
\Bigg[
-j\pi m_n^{\mathrm t}
\cos(\tilde{\psi}_q)
\left(
\bar{\mathbf f}_n^T\mathbf u_{q,n}
\right)^p
+
p
\left(
\bar{\mathbf f}_n^T\mathbf u_{q,n}
\right)^{p-1}
\zeta_{q,n}
\Bigg],
\qquad
q\in\{k,(e,m)\},\label{channel_array_rotation_derivative}
\\
\frac{\partial\mathcal G}
{\partial\phi_{\mathrm{arr}}^c}
={}&
\frac{2\gamma_c}{N_t\ln 2}
\sum_{k=1}^{K}
a_k
\frac{
\operatorname{Re}
\left\{
\varsigma_k^*
\left(
\frac{\partial\mathbf h_k^H}
{\partial\phi_{\mathrm{arr}}^c}
\mathbf v
\right)
\right\}
}{
1+\frac{\gamma_c}{N_t}
|\varsigma_k|^2
}
-
\frac{2\gamma_e}{N_t\ln 2}
\sum_{m=1}^{M}
b_m
\frac{
\operatorname{Re}
\left\{
\varsigma_{e,m}^*
\left(
\frac{\partial\mathbf h_{e,m}^H}
{\partial\phi_{\mathrm{arr}}^c}
\mathbf v
\right)
\right\}
}{
1+\frac{\gamma_e}{N_t}
|\varsigma_{e,m}|^2
},\label{grad_array_rotation}
\\
\frac{\partial\varsigma_q}
{\partial\varphi_n^c}
={}&
\beta_q^*
\left[
\mathbf a_t(\tilde{\psi}_q)
\right]_n^*
\sqrt{G_0}\,p
\left(
\bar{\mathbf f}_n^T\mathbf u_{q,n}
\right)^{p-1}
\left(
\dot{\mathbf f}_n^T
\mathbf R^T(\phi_{\mathrm{arr}}^c)
\mathbf u_{q,n}
\right)
v_n,
\qquad
q\in\{k,(e,m)\},
\\
\frac{\partial\mathcal G}
{\partial\varphi_n^c}
={}&
\frac{2\gamma_c}{N_t\ln 2}
\sum_{k=1}^{K}
a_k
\frac{
\operatorname{Re}
\left\{
\varsigma_k^*
\frac{\partial\varsigma_k}
{\partial\varphi_n^c}
\right\}
}{
1+\frac{\gamma_c}{N_t}
|\varsigma_k|^2
}
-
\frac{2\gamma_e}{N_t\ln 2}
\sum_{m=1}^{M}
b_m
\frac{
\operatorname{Re}
\left\{
\varsigma_{e,m}^*
\frac{\partial\varsigma_{e,m}}
{\partial\varphi_n^c}
\right\}
}{
1+\frac{\gamma_e}{N_t}
|\varsigma_{e,m}|^2
},
\qquad
n=1,\ldots,N_t,\label{grad_element_rotation}
\end{align}
\end{figure*}
The derivatives in \eqref{channel_array_rotation_derivative} and
\eqref{grad_element_rotation} are evaluated on the active branch
$\bar{\mathbf f}_n^T\mathbf u_{q,n}>0$ of the truncated directional
pattern and become zero on the inactive branch.

\textit{3) Joint Product-Space Optimization:}
To jointly update the analog beamformer and the RA configuration,
we first construct a product space that accommodates their distinct
geometric structures. Since the array-level and element-level
rotations have different feasible ranges, the RA configuration is
normalized as follows:
\begin{equation}
\widetilde{\Phi}_c
\triangleq
\mathbf D_c^{-1}\Phi_c
=
\begin{bmatrix}
\displaystyle
\frac{\phi_{\mathrm{arr}}^c}{\phi_{\max}}
&
\displaystyle
\frac{\varphi_1^c}{\varphi_{\max}}
&
\cdots
&
\displaystyle
\frac{\varphi_{N_t}^c}{\varphi_{\max}}
\end{bmatrix}^{T},
\label{normalized_RA_configuration}
\end{equation}
where
$\mathbf D_c
\triangleq
\operatorname{diag}
\left(
\phi_{\max},
\varphi_{\max}\mathbf I_{N_t}
\right)$.
Accordingly, the bounded-rotation constraints are equivalently
represented by
\begin{equation}
\widetilde{\Phi}_c\in\mathcal B,
\qquad
\mathcal B=[-1,1]^{N_t+1},
\label{normalized_rotation_set}
\end{equation}
and the physical RA configuration is recovered as
$\Phi_c=\mathbf D_c\widetilde{\Phi}_c$. By the chain rule, the
gradient with respect to the normalized RA configuration is
\begin{equation}
\nabla_{\widetilde{\Phi}_c}\mathcal G
=
\mathbf D_c\nabla_{\Phi_c}\mathcal G.
\label{normalized_rotation_gradient}
\end{equation}

Combining the complex circle manifold $\mathcal M$ in
\eqref{communication_complex_circle_manifold} and the normalized
rotation set $\mathcal B$, the feasible product space is constructed
as follows:
\begin{equation}
\mathcal P
=
\mathcal M\times\mathcal B
=
\{
(\mathbf v,\widetilde{\Phi}_c)
\;|\;
|v_n|=1,\ \forall n,\ 
\widetilde{\Phi}_c\in\mathcal B
\}.
\label{communication_product_space}
\end{equation}
At $(\mathbf v,\widetilde{\Phi}_c)\in\mathcal P$, the local
product search space is given by
\begin{equation}
\mathcal S_{(\mathbf v,\widetilde{\Phi}_c)}
=
\mathcal T_{\mathbf v}\mathcal M
\times
\mathbb R^{N_t+1}.
\label{local_product_search_space}
\end{equation}
% {\color{red}The first component accounts for the geometry induced by the CM
% constraint, whereas feasibility of the second component is enforced
% through projection onto $\mathcal B$.}

Based on \eqref{Riemannian_grad_w} and
\eqref{normalized_rotation_gradient}, we develop a CGA algorithm
over $\mathcal P$. Given an initial feasible point
$(\mathbf v^{(0)},\widetilde{\Phi}_c^{(0)})$, each iteration
consists of three steps: 1) calculation of the joint product-space
gradient, 2) update of the feasible solution, and 3) determination
of the step size using Armijo backtracking.

\textit{a) Joint product-space gradient:}
At the $r$-th iteration, the projected gradient for the normalized RA configuration is defined as follows:
\begin{equation}
\mathbf g_{\Phi,r}
\triangleq
\operatorname{Proj}_{\mathcal B}
(
\widetilde{\Phi}_c^{(r)}
+
\nabla_{\widetilde{\Phi}_c}\mathcal G^{(r)}
)
-
\widetilde{\Phi}_c^{(r)},
\label{rotation_projected_gradient}
\end{equation}
where $\operatorname{Proj}_{\mathcal B}(\cdot)$ denotes the
Euclidean projection onto $\mathcal B$. Thus, the joint product-space
gradient is given by
\begin{equation}
\mathbf g^{(r)}
=
[\mathbf g_{v,r};
\mathbf g_{\Phi,r}],
\label{joint_product_gradient}
\end{equation}
where $\mathbf g_{v,r}
\triangleq
\operatorname{grad}_{\mathbf v}
\mathcal G
(
\mathbf v^{(r)},\Phi_c^{(r)}
)$ is the Riemannian gradient on $\mathcal M$.

\textit{b) Update of the feasible solution:}
The current product-space gradient can be directly used
as the ascent direction. However, utilizing this ascent direction 
exploits the local ascent direction and discards the
search information accumulated in previous iterations. 
To address this issue, we adopt a conjugate ascent direction, which combines the
current gradient with the transported direction from the previous
iteration. Specifically, the joint conjugate ascent direction is
given by
\begin{equation}
\begin{aligned}
\mathbf d^{(r)}
&=[
\mathbf d_v^{(r)};
\mathbf d_{\Phi}^{(r)}]\\
&=
\begin{bmatrix}
\mathbf g_{v,r}
+
\sigma_r
\operatorname{Tran}_{\mathbf v^{(r)}
\leftarrow\mathbf v^{(r-1)}}
(
\mathbf d_v^{(r-1)}
);
\\
\mathbf g_{\Phi,r}
+
\sigma_r
\operatorname{Tran}_{\widetilde{\Phi}_c^{(r)}
\leftarrow\widetilde{\Phi}_c^{(r-1)}}
(
\mathbf d_{\Phi}^{(r-1)}
)
\end{bmatrix},
\end{aligned}
\label{joint_conjugate_direction}
\end{equation}
where $\sigma_r$ is the conjugate parameter which is defined in
\eqref{product_conjugate_parameter};
\begin{figure*}[!b]
\hrulefill
\begin{align}
\sigma_r
={}&
\left[
\frac{
\operatorname{Re}
\left\{
\mathbf g_{v,r}^{H}
\left[
\mathbf g_{v,r}
-
\left(
\mathbf g_{v,r-1}
-
\operatorname{Re}
\{
\mathbf g_{v,r-1}
\odot
(\mathbf v^{(r)})^*
\}
\odot
\mathbf v^{(r)}
\right)
\right]
\right\}
+
\mathbf g_{\Phi,r}^{T}
(
\mathbf g_{\Phi,r}
-
\mathbf g_{\Phi,r-1}
)
}{
\begin{gathered}
-
\operatorname{Re}
\left\{
\left[
\mathbf d_v^{(r-1)}
-
\operatorname{Re}
\{
\mathbf d_v^{(r-1)}
\odot
(\mathbf v^{(r)})^*
\}
\odot
\mathbf v^{(r)}
\right]^{H}
\left[
\mathbf g_{v,r}
-
\left(
\mathbf g_{v,r-1}
-
\operatorname{Re}
\{
\mathbf g_{v,r-1}
\odot
(\mathbf v^{(r)})^*
\}
\odot
\mathbf v^{(r)}
\right)
\right]
\right\}
\\[-1mm]
-
(\mathbf d_{\Phi}^{(r-1)})^T
(
\mathbf g_{\Phi,r}
-
\mathbf g_{\Phi,r-1}
)
\end{gathered}
}
\right]^+
\label{product_conjugate_parameter}
\end{align}
\end{figure*}
$\operatorname{Tran}_{\mathbf v^{(r)}
\leftarrow\mathbf v^{(r-1)}}(\cdot)$ and
$\operatorname{Tran}_{\widetilde{\Phi}_c^{(r)}
\leftarrow\widetilde{\Phi}_c^{(r-1)}}(\cdot)$ transfer the previous search
directions to the current local product search space, respectively, as follows:
\begin{subequations}
\label{product_vector_transport}
\begin{align}
&
\begin{aligned}
&\operatorname{Tran}_{\mathbf v^{(r)}
\leftarrow\mathbf v^{(r-1)}}
\bigl(\mathbf d_v^{(r-1)}\bigr)
\\
&\qquad=
\mathbf d_v^{(r-1)}
-
\operatorname{Re}
\bigl\{
\mathbf d_v^{(r-1)}
\odot
(\mathbf v^{(r)})^*
\bigr\}
\odot
\mathbf v^{(r)}
\end{aligned}
\label{beamforming_vector_transport}
\\[1mm]
&
\operatorname{Tran}_{\widetilde{\Phi}_c^{(r)}
\leftarrow\widetilde{\Phi}_c^{(r-1)}}
\bigl(\mathbf d_{\Phi}^{(r-1)}\bigr)
=
\mathbf d_{\Phi}^{(r-1)}.
\label{rotation_vector_transport}
\end{align}
\end{subequations}

Given the conjugate ascent direction in
\eqref{joint_conjugate_direction}, the update on the
local product search space is
\begin{equation}
\begin{aligned}
\bigl[
\widehat{\mathbf v}^{(r)};
\widehat{\Phi}_c^{(r)}
\bigr]
&=
\bigl[
\mathbf v^{(r)};
\widetilde{\Phi}_c^{(r)}
\bigr]
+
\alpha_r
\bigl[
\mathbf d_v^{(r)};
\mathbf d_{\Phi}^{(r)}
\bigr]
\\
&=
\bigl[
\mathbf v^{(r)}
+
\alpha_r\mathbf d_v^{(r)};
\widetilde{\Phi}_c^{(r)}
+
\alpha_r\mathbf d_{\Phi}^{(r)}
\bigr],
\end{aligned}
\label{product_tentative_update}
\end{equation}
where $\alpha_r>0$ is the step size. Generally, the update in
\eqref{product_tentative_update} does not necessarily remain in the
feasible product space $\mathcal P$. In such cases, a retraction-projection
operation is applied to obtain the next feasible iterate as follows:
\begin{equation}
\begin{aligned}
[\mathbf v^{(r+1)};\widetilde{\Phi}_c^{(r+1)}]
&=
\bigl[
\operatorname{Ret}_{\mathcal M}
\bigl(\widehat{\mathbf v}^{(r)}\bigr);
\operatorname{Proj}_{\mathcal B}
\bigl(\widehat{\Phi}_c^{(r)}\bigr)
\bigr]
\\
&=
\left[
\widehat{\mathbf v}^{(r)}
\odot
\frac{1}{|\widehat{\mathbf v}^{(r)}|};
\operatorname{Proj}_{\mathcal B}
\bigl(\widehat{\Phi}_c^{(r)}\bigr)
\right],
\end{aligned}
\label{product_feasible_update}
\end{equation}
where $\operatorname{Ret}_{\mathcal M}(\cdot)$ and
$\operatorname{Proj}_{\mathcal B}(\cdot)$ denote the retraction onto
the complex circle manifold $\mathcal M$ and the Euclidean projection
onto the box-constrained set $\mathcal B$, respectively.

\textit{c) Armijo step-size update:}
Using the Armijo line search strategy, the step size
is adaptively determined by requiring
\begin{equation}
\begin{aligned}
\mathcal G
(
\mathbf v^{(r+1)},
\mathbf D_c\widetilde{\Phi}_c^{(r+1)}
)
\geq{}&
\mathcal G
(
\mathbf v^{(r)},
\mathbf D_c\widetilde{\Phi}_c^{(r)}
)
+
c_{\mathrm A}\tau^{q_r}\widehat{\alpha}
\widehat d_{\mathcal G}^{(r)},
\end{aligned}
\label{product_Armijo_condition}
\end{equation}
where $\tau\in(0,1)$ is the step size coefficient,
$\widehat{\alpha}>0$ is the initial trial step size,
$c_{\mathrm A}\in(0,1)$ is the sufficient-ascent parameter, and
$q_r$ is the number of line searches. The final step size for the $r$-th
iteration is given by 
$\alpha_r
=\tau^{q_r}\widehat{\alpha}$.
$\widehat d_{\mathcal G}^{(r)}$ represents the directional derivative, which is defined as follows:
\begin{equation}
\widehat d_{\mathcal G}^{(r)}
=\operatorname{Re}
\{
\mathbf g_{v,r}^{H}\mathbf d_v^{(r)}
\}
+
(
\nabla_{\widetilde{\Phi}_c}\mathcal G^{(r)}
)^T
\operatorname{Proj}_{\mathcal T_{\mathcal B}
(\widetilde{\Phi}_c^{(r)})}
(
\mathbf d_{\Phi}^{(r)}
),
\label{product_directional_derivative}
\end{equation}
where
$\operatorname{Proj}_{\mathcal T_{\mathcal B}
(\widetilde{\Phi}_c^{(r)})}(\cdot)$ denotes the projection onto the tangent cone of
$\mathcal B$ at $\widetilde{\Phi}_c^{(r)}$.
Therefore, the analog beamformer and all RA rotation variables are
jointly updated using the common step size $\alpha_r$.

The iterations terminate when
$
\sqrt{
\|
\mathbf g_{v,r}
\|_2^2
+
\|
\mathbf g_{\Phi,r}
\|_2^2
}
\leq
\epsilon_{\mathrm p}$
,
or the maximum number of iterations is reached.
After convergence, the analog beamformer and the communication-stage
RA configuration are recovered as follows:
\begin{equation}
\mathbf w^\star
=
\frac{\mathbf v^\star}{\sqrt{N_t}},
\qquad
\Phi_c^\star
=
\mathbf D_c\widetilde{\Phi}_c^\star.
\label{PSJO_recovered_solution}
\end{equation}

The proposed PSJO algorithm for solving the communication stage
problem~\eqref{P3} is summarized in
Algorithm~\ref{alg:overall_algorithm}. The Armijo line search yields
a nondecreasing sequence of objective values. Since $\mathcal G$ is
continuous on the compact feasible set $\mathcal P$, this sequence is
bounded above and thus converges.

\subsection{Overall Algorithm}

\begin{algorithm}[!t]
\caption{PSJO Algorithm for Communication-Stage RA and Beamforming Design}
\label{alg:overall_algorithm}
\begin{algorithmic}[1]
\State \textbf{Input:}
$\{\mathbf h_k\}_{k=1}^{K}$,
$\{\mathbf h_{e,m}\}_{m=1}^{M}$,
$\beta_c$, $\widehat{\alpha}$, $\tau$,
$c_{\mathrm A}$, $\epsilon_{\mathrm p}$, and $I_{\mathrm p}$.
\State Initialize $\mathbf v^{(0)}\in\mathcal M$ and
$\widetilde{\Phi}_c^{(0)}\in\mathcal B$, and set $r=0$.
\State Compute $a_k$ and $b_m$ according to
\eqref{communication_smoothing_weights}.
\State Compute $\mathbf g_{v,0}$ and $\mathbf g_{\Phi,0}$
according to \eqref{Riemannian_grad_w} and
\eqref{rotation_projected_gradient}.
\Repeat
    \If{$r=0$}
        \State Set
        $\mathbf d^{(0)}\leftarrow\mathbf g^{(0)}$.
    \Else
        \State Compute the vector transports according to
        \eqref{product_vector_transport}.
        \State Compute $\sigma_r$ according to
        \eqref{product_conjugate_parameter}.
        \State Compute $\mathbf d^{(r)}$ according to
        \eqref{joint_conjugate_direction}.
    \EndIf
    \State Determine $q_r$ satisfying
    \eqref{product_Armijo_condition}, and set
    $\alpha_r\leftarrow\tau^{q_r}\widehat{\alpha}$.
    \State Compute the tentative point according to
    \eqref{product_tentative_update}.
    \State Update $\mathbf v^{(r+1)}$ and
    $\widetilde{\Phi}_c^{(r+1)}$ according to
    \eqref{product_feasible_update}.
    \State Set $r\leftarrow r+1$.
    \State Compute $a_k$ and $b_m$ according to
    \eqref{communication_smoothing_weights}.
    \State Compute $\mathbf g_{v,r}$ and $\mathbf g_{\Phi,r}$
    according to \eqref{Riemannian_grad_w} and
    \eqref{rotation_projected_gradient}.
\Until{
    $\sqrt{
    \|\mathbf g_{v,r}\|_2^2+
    \|\mathbf g_{\Phi,r}\|_2^2
    }
    \leq\epsilon_{\mathrm p}$
    or $r\geq I_{\mathrm p}$
}
\State \textbf{Output:}
$\mathbf w^\star=\mathbf v^{(r)}/\sqrt{N_t}$ and
$\Phi_c^\star=\mathbf D_c\widetilde{\Phi}_c^{(r)}$.
\end{algorithmic}
\end{algorithm}

\textit{1) Complexity of the sensing stage:}
The sensing-stage complexity is dominated by the repeated evaluation
of $\mathcal J(\Phi_s)$ during the cyclic coordinate search in
Algorithm~\ref{alg:sensing_coordinate_search}. Each search cycle
examines $L_{\mathrm{arr}}+1$ array-level candidates and
$N_t(L_{\mathrm{ele}}+1)$ element-level candidates. Every candidate
is evaluated over the $M_1+1$ angular samples in $\mathcal T$.
After precomputing the receive-array quantities, computing
$\Psi(\theta_{m_1};\Phi_s)$ at one angular sample requires
$\mathcal O(N_t)$ operations. Therefore, the sensing-stage
complexity is
$
\mathcal O
\left(
I_{\mathrm s}
(M_1+1)N_t
\left[
L_{\mathrm{arr}}+1
+
N_t(L_{\mathrm{ele}}+1)
\right]
\right),
$
where $I_{\mathrm s}$ denotes the number of coordinate-search
cycles.

\textit{2) Complexity of the communication stage:}
The computational cost of the PSJO algorithm consists of the
product-space update and the Armijo line search. For the
product-space update, directly evaluating the analog-beamforming
gradient in \eqref{grad_w} requires
$\mathcal O((K+M)N_t^2)$ operations. The RA-rotation gradients from
\eqref{grad_array_rotation} to \eqref{grad_element_rotation}, together
with the conjugate direction and the retraction-projection update,
require $\mathcal O((K+M)N_t)$ operations. Therefore, the complexity
of the product-space update is dominated by
$\mathcal O((K+M)N_t^2)$.
The Armijo line search evaluates $\mathcal G$ once for each 
step, with a cost of $\mathcal O((K+M)N_t)$. Let
$\overline{Q}_{\mathrm A}$ denote the average number of steps
per iteration. The corresponding line-search complexity is
$\mathcal O(\overline{Q}_{\mathrm A}(K+M)N_t)$.
Combining these two parts, the overall complexity of the PSJO
algorithm is
$
\mathcal O
(
I_{\mathrm p}
[
(K+M)N_t^2
+
\overline{Q}_{\mathrm A}(K+M)N_t
]
),
$
where $I_{\mathrm p}$ denotes the number of PSJO iterations.

\section{Numerical Results}
In this section, numerical results are provided to evaluate the performance of the proposed sensing-aided secure communication design.
In the simulations, the carrier frequency is set to $28$~GHz. The
BS is equipped with $N_t=8$ transmit antennas and $N_r=16$ receive
antennas. Unless otherwise specified, the communication and sensing
transmit powers are set to $P_t=20$~dBm and $P_s=35$~dBm,
respectively. The noise powers in the communication and sensing
stages are set to $\sigma_c^2=\sigma_s^2=\sigma_e^2=-90$~dBm. The RCS of the
eavesdropper is set to $\alpha=7$~dBsm, and the antenna directivity
factor is set to $p=1$.
The BS serves $K=3$ legitimate users. For each Monte Carlo
realization, the user distances and azimuth angles are independently
and uniformly drawn from $[30,50]$~m and
$(-90^\circ,90^\circ)$, respectively. The eavesdropper is located at
$(30~\mathrm{m},30^\circ)$. The initial angular search region is set
to $\Theta_1=[20^\circ,40^\circ]$. The number of sensing beams is $L=16$.
The angular traversal set over $\Theta_1$ contains $M_1+1=401$
uniformly spaced points, while the CRB-based uncertainty region
$\Theta_2$ is uniformly discretized into $M=51$ points.
The maximum array-level and element-level rotation angles are both
set to $15^\circ$. The candidate sets for the array-level and
element-level rotations contain $L_{\mathrm{arr}}+1=101$ and
$L_{\mathrm{ele}}+1=101$ uniformly spaced points, respectively.
The smoothing parameter is set to $\beta_c=100$. For the Armijo line
search, the initial trial step, contraction factor, and
sufficient-ascent parameter are set to
$\widehat{\alpha}=1$, $\tau=0.5$, and
$c_{\mathrm A}=10^{-4}$, respectively. The convergence tolerances
are set to $\epsilon_{\mathrm c}=\epsilon_{\mathrm p}=10^{-6}$, and
the maximum number of PSJO iterations is $I_{\mathrm p}=200$. 
All results are averaged over $200$ independent Monte Carlo realizations.
The empirical MSE of the MLE is
evaluated using $20000$ independent noise realizations.

For comparison, the following benchmark schemes are considered:
1) $\textbf{FPA:}$ Fixed-position array with analog beamforming.
2) $\textbf{GRA:}$ Globally rotatable array with analog beamforming, where only array-level rotation is enabled.
3) $\textbf{ERA:}$ Element-rotatable array with analog beamforming, where only element-wise rotation is enabled.
4) $\textbf{TRA-PE:}$ Two-level rotatable array with analog beamforming using
    only the point estimate $\hat{\theta}_e$, without considering the CRB-based angular
    uncertainty region.
5) $\textbf{TRA-MRT:}$ Two-level rotatable array with MRT-inspired analog beamforming, optimized to maximize the minimum legitimate-user channel gain without using eavesdropper information.
For sensing performance, 6) $\textbf{Lower bound:}$
An analytical CRB lower bound is obtained by dropping the non-positive
correlation term in $\Psi(\theta;\Phi_s)$ and replacing the
rotation-dependent directional-gain and angular-sensitivity terms with
their respective maximum values.
\subsection{Simulation Results}
\begin{figure}[!t]
    \centering
    \captionsetup{font=small}
    \begin{subfigure}[t]{0.48\linewidth}
        \centering
        \includegraphics[width=\linewidth]{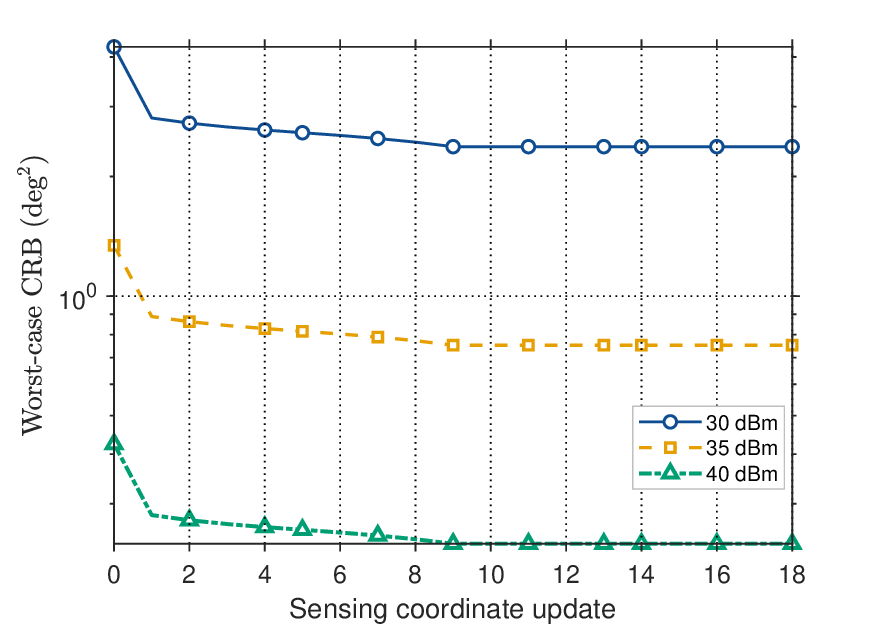}
        \caption{Algorithm 1}
        \label{fig:sensing_convergence}
    \end{subfigure}
    \hfill
    \begin{subfigure}[t]{0.48\linewidth}
        \centering
        \includegraphics[width=\linewidth]{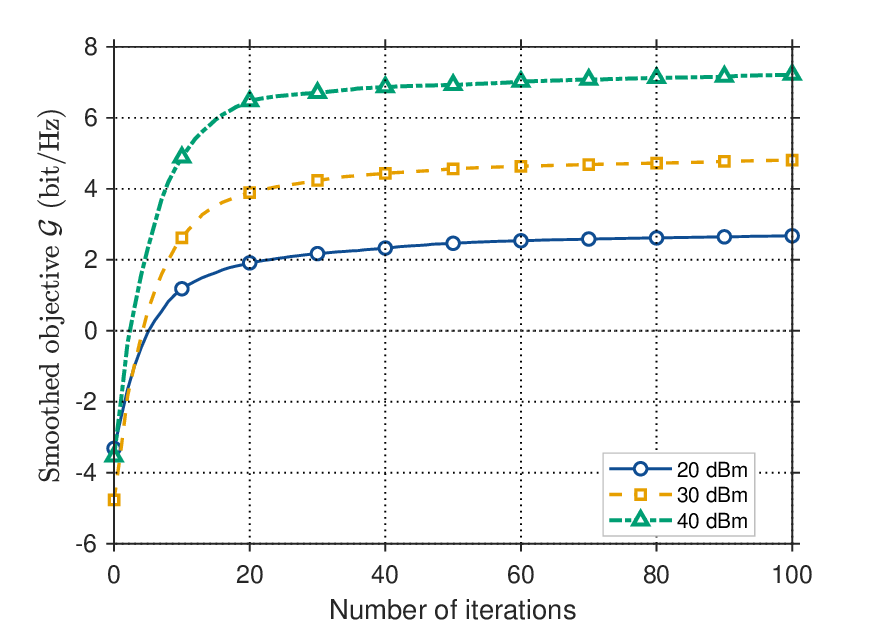}
        \caption{Algorithm 2}
        \label{fig:communication_convergence}
    \end{subfigure}
    \caption{\justifying Convergence behavior of the proposed algorithms.}
    \label{fig:convergence}
\end{figure}

Fig.~\ref{fig:convergence} illustrates the convergence behavior of
the proposed algorithms under different transmit-power levels. As
shown in Fig.~\ref{fig:convergence}(a), the worst-case CRB decreases
monotonically with the coordinate updates, with most of the reduction
achieved within the first few updates. All curves become nearly
stable after approximately ten updates. Moreover, increasing $P_s$
leads to a lower converged CRB owing to the improved sensing SNR. In
Fig.~\ref{fig:convergence}(b), the smoothed objective $\mathcal G$
increases monotonically for all considered $P_t$. The main
improvement occurs within the first $20$ iterations, after which the
curves gradually approach stationary values. A higher $P_t$ also
yields a larger converged objective value. These results are
consistent with the established monotonicity of the two algorithms
and confirm their convergence within a moderate number of
iterations.

\begin{figure}[!t]
    \centering
    \includegraphics[width=\linewidth]{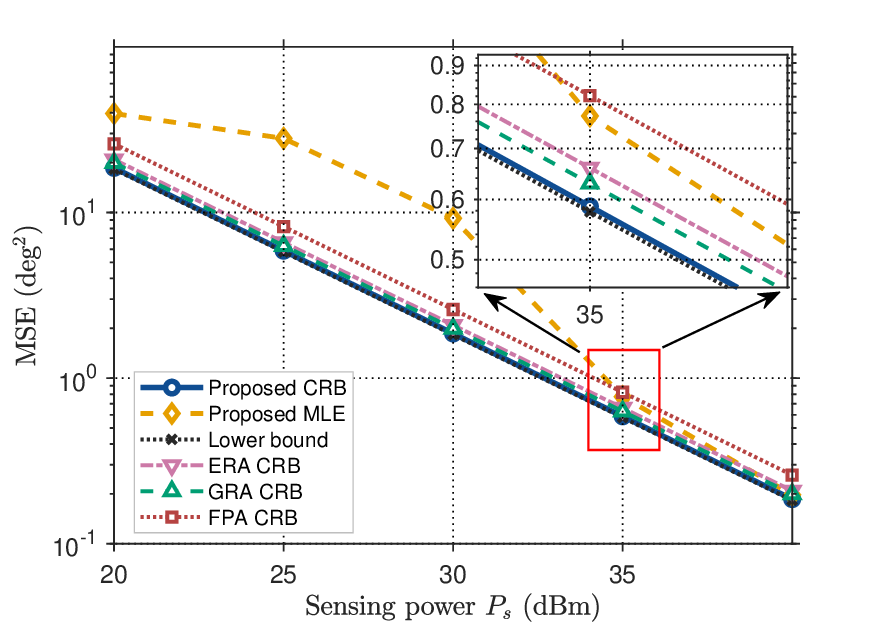}
    \captionsetup{font={small}}
    \caption{\justifying MSE of eavesdropper AoD estimation versus $P_s$.}
    \label{fig:mse}
\end{figure}

Fig.~\ref{fig:mse} evaluates the eavesdropper AoD estimation
performance in terms of the CRB and the empirical MSE of the MLE
versus the sensing power $P_s$. As $P_s$ increases, the CRBs of all
schemes decrease owing to the improved sensing SNR. The proposed
scheme achieves the lowest CRB among the considered array
architectures and remains close to the lower bound.
Compared with the GRA, ERA, and FPA schemes, jointly optimizing the
array-level and element-level rotations further reduces the CRB. At
high sensing powers, e.g., $P_s\geq 35~\mathrm{dBm}$, the MSE of the
MLE approaches the proposed CRB, showing that the CRB is tight in
this regime. At lower sensing powers, the pronounced gap between the
MLE and the CRB is consistent with the threshold effect of nonlinear
parameter estimation.

\begin{figure}[!t]
    \centering
    \includegraphics[width=\linewidth]{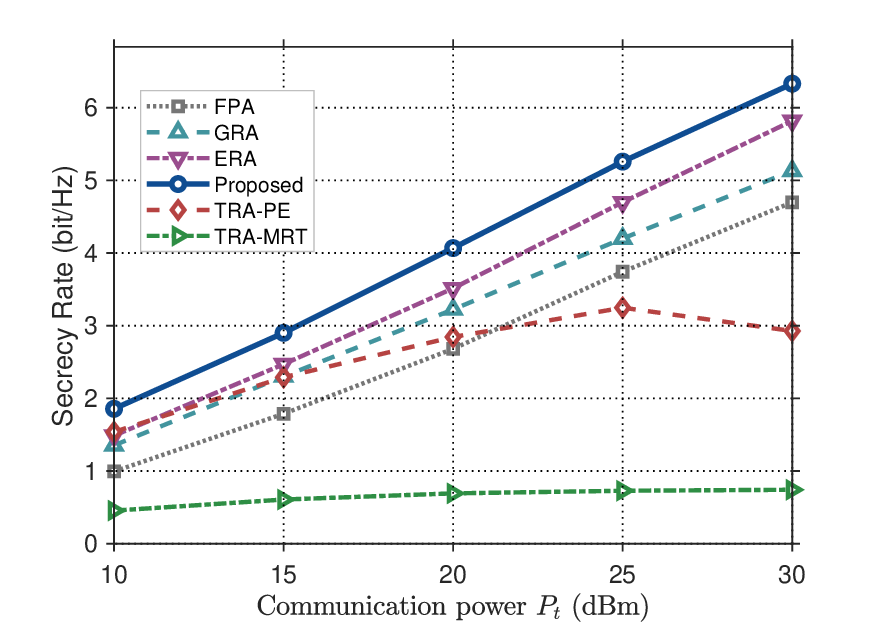}
    \captionsetup{font={small}}
    \caption{\justifying Secrecy rates of different schemes versus $P_t$.}
    \label{fig:Pt}
\end{figure}

Fig.~\ref{fig:Pt} shows the secrecy rates of different schemes versus
$P_t$. The secrecy rates of the proposed, ERA, GRA, and FPA schemes
increase steadily with $P_t$, with the proposed scheme performing
best over the entire power range. ERA consistently outperforms GRA,
while joint array-level and element-level rotations provide a further
gain. TRA-MRT changes only slightly with $P_t$ because it does not
use eavesdropper information to suppress leakage. The secrecy rate of
TRA-PE increases up to $P_t=25~\mathrm{dBm}$ but decreases at
$30~\mathrm{dBm}$. Since TRA-PE treats $\hat{\theta}_e$ as the true
eavesdropper direction, its leakage suppression becomes less reliable
as the effect of the estimation error grows with $P_t$.

\begin{figure}[!t]
    \centering
    \includegraphics[width=\linewidth]{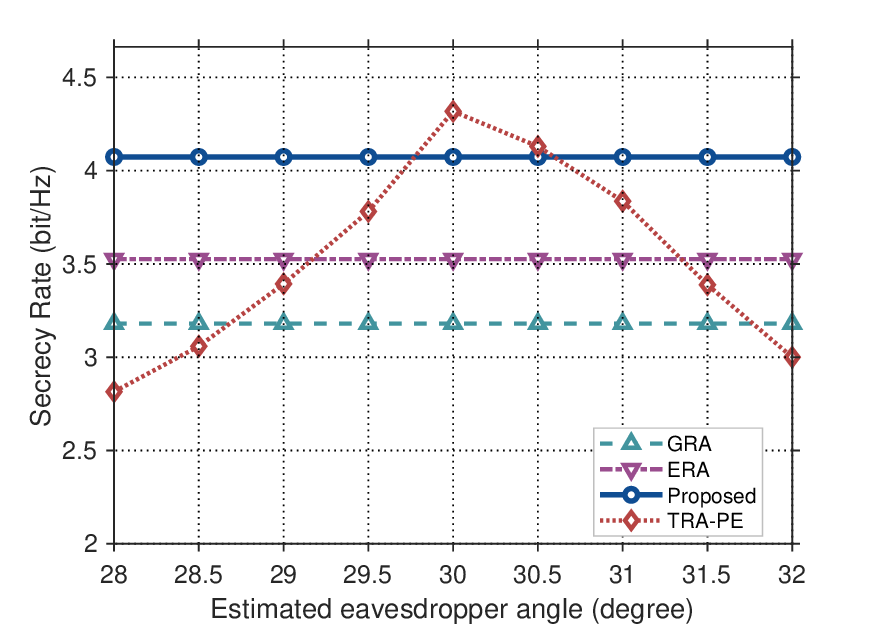}
    \captionsetup{font={small}}
    \caption{\justifying Performance comparison between the proposed and
the Non-robust scheme.}
    \label{fig:beampattern}
\end{figure}

Fig.~\ref{fig:beampattern} compares the robustness of different
schemes against eavesdropper-angle estimation errors. The true
eavesdropper angle is fixed at $\theta_e=30^\circ$, while
$\hat{\theta}_e$ varies over $\Theta_2=[28^\circ,32^\circ]$. The
proposed scheme maintains an almost constant secrecy rate because it
accounts for the entire uncertainty interval. The ERA and GRA curves
are also stable but remain below the proposed scheme because of their
restricted rotation capabilities. TRA-PE attains its highest secrecy
rate when $\hat{\theta}_e=\theta_e$ and slightly outperforms the
proposed scheme at this point. Its secrecy rate drops rapidly as the
estimation error increases. The point-estimate design therefore
offers a higher rate when the estimate is exact, whereas the
robust schemes avoid large losses over the uncertainty interval.

\begin{figure}[!t]
    \centering
    \includegraphics[width=\linewidth]{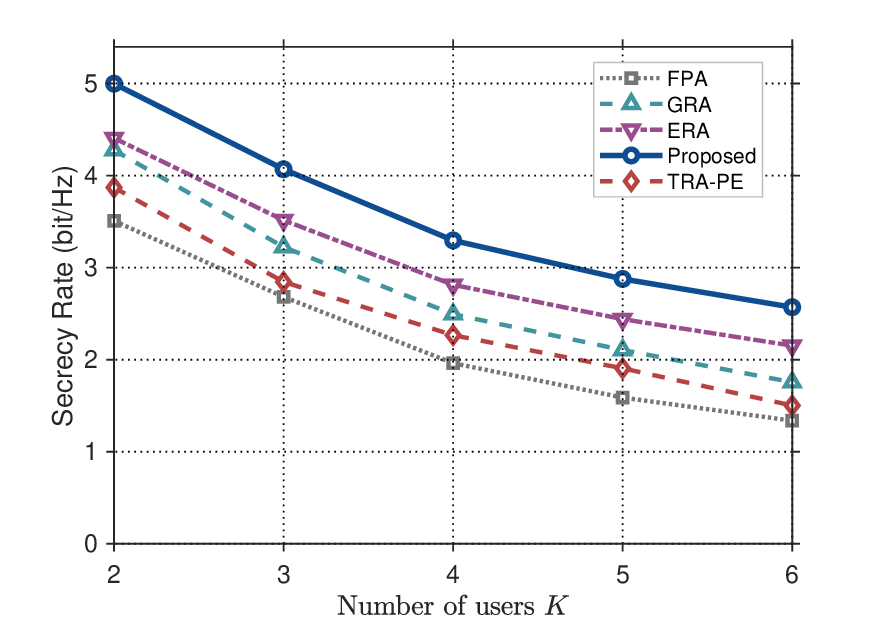}
    \captionsetup{font={small}}
    \caption{\justifying Secrecy rates of different schemes versus $K$.}
    \label{fig:K}
\end{figure}

Fig.~\ref{fig:K} shows the secrecy rates of different schemes versus
$K$. The secrecy rates decrease as the number of legitimate users
increases because the single CM analog beamformer must balance more
user channels, making the worst-user bottleneck more restrictive.
The proposed scheme maintains the highest secrecy rate over the
entire range. The ordering of the remaining schemes is unchanged,
with ERA followed by GRA, TRA-PE, and FPA. The performance gap between
the proposed and single-level RA schemes persists as $K$ increases.

\begin{figure}[!t]
    \centering
    \includegraphics[width=\linewidth]{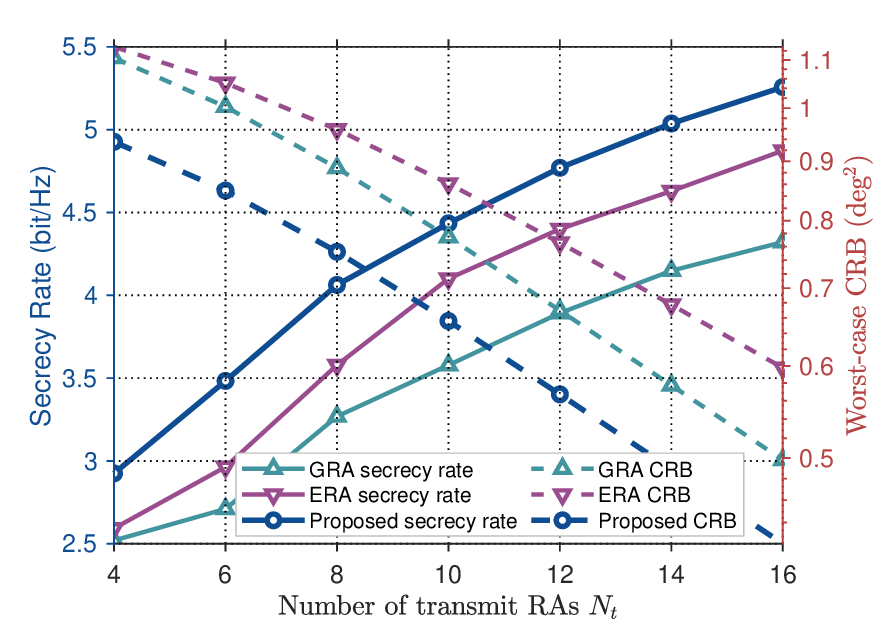}
    \captionsetup{font={small}}
    \caption{\justifying  Secrecy rates and CRB of the proposed scheme versus number of RAs.}
    \label{fig:numberRA}
\end{figure}
Fig.~\ref{fig:numberRA} illustrates the secrecy rates and the
corresponding worst-case CRBs versus the number of transmit RAs
$N_t$. As $N_t$ increases, the secrecy rates rise steadily, whereas
the worst-case CRBs decline accordingly. Additional RAs provide
higher array gain for eavesdropper AoD estimation and more spatial
DoFs for beamforming, allowing stronger legitimate reception and
more effective leakage suppression. These results demonstrate that increasing the number of RAs is 
beneficial for enhancing both sensing accuracy and secrecy performance.

\begin{figure}[!t]
    \centering
    \includegraphics[width=\linewidth]{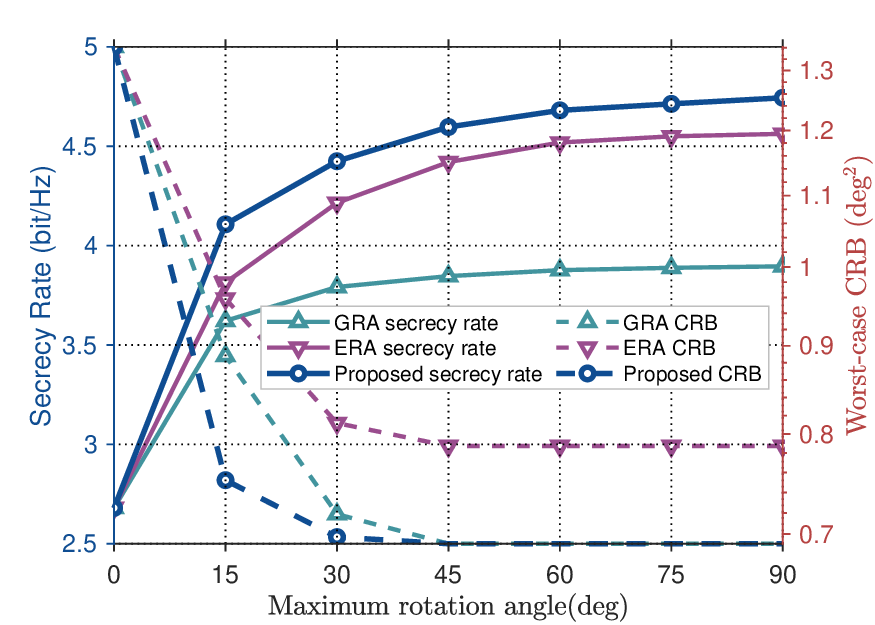}
    \captionsetup{font={small}}
    \caption{\justifying  Secrecy rates and CRB of the proposed scheme versus maximum rotation angle.}
    \label{fig:anglelimit}
\end{figure}

Fig.~\ref{fig:anglelimit} shows the secrecy rate and worst-case CRB
versus the maximum allowable rotation angle, where
$\phi_{\max}=\varphi_{\max}$. When this angle is zero, all three RA
architectures reduce to the same fixed array and achieve identical
performance. Increasing the feasible rotation range raises the
secrecy rate and reduces the CRB by allowing more favorable
orientations for eavesdropper sensing and secure beamforming. Most
of the improvement occurs below approximately $45^\circ$. Beyond
this point, the optimized rotations already lie within the feasible
range, and further enlargement yields little change. For nonzero
rotation limits, GRA achieves a lower CRB than ERA, whereas ERA
provides a higher secrecy rate than GRA. Their joint use gives the
proposed scheme the highest secrecy rate and the lowest CRB level.
\section{Conclusion}
In this paper, we studied sensing-aided secure multicast transmission
in a two-level RA-enabled ISAC system. In the sensing stage, the
array-level and element-level rotations were optimized to minimize the
worst-case AoD CRB. The MLE and the resulting CRB determined the
center and width of the eavesdropper angular uncertainty region,
respectively. A cyclic coordinate search algorithm was developed for
the sensing-stage design. In the communication stage, the CM analog
beamformer and RA configuration were jointly optimized to maximize
the worst-case secrecy rate over this region. The resulting problem
was solved using smooth approximation and the proposed PSJO
framework. Numerical results verified the convergence of both
algorithms. The proposed design approached the analytical CRB lower bound
and achieved higher secrecy rates than the considered
benchmarks. It also maintained stable performance under
eavesdropper-angle estimation errors.
\appendices

\section{Proof of Lemma 1}

For ease of exposition, we first vectorize the received sensing signal matrix \(\mathbf Y_s\) as follows:
\begin{equation}
\bar {\mathbf y}
=
\operatorname{vec}(\mathbf Y_s)
=
\beta_s \operatorname{vec}{\mathbf q(\theta_e)}+\operatorname{vec}(\mathbf N_s),
\end{equation}
where
$\mathbf q(\theta)
=\mathbf a_r(\theta)\mathbf b^H(\theta)\mathbf X_s.$
For fixed parameter set \(\bm \chi \triangleq \{\Re ({\beta}_s),\Im ({\beta}_s),\theta\}\), the likelihood function associated with the vector of variables $\bm \chi$ is
\begin{equation}
\mathcal L(\bar{\mathbf y};\bm \chi)
=
\frac{1}{(\pi\sigma_s^2)^{LN_r}}
\exp\!\left(
-\frac{1}{\sigma_s^2}
\left\|
\bar {\mathbf y}-\beta_s\operatorname{vec}(\mathbf q(\theta))
\right\|^2
\right).
\end{equation}
To maximize the likelihood function, it is equivalent to solve
\begin{equation}\label{likelihood}
({\theta}^*_e,{\beta}^*_s)
=
\arg\min_{\theta\in\Theta_1,\beta_s}
\left\|
\bar {\mathbf y}-\beta_s\operatorname{vec}(\mathbf q(\theta))
\right\|^2.
\end{equation}
Given any \(\theta\), the optimal problem related to \(\beta_s\) can be regarded as a linear least squares problem, and its optimum is given by
\begin{equation} \label{beta*}
{\beta}^*_s
=
\frac{(\operatorname{vec}(\mathbf q(\theta)))^H \bar {\mathbf y}}{\|\operatorname{vec}(\mathbf q(\theta))\|^2}.
\end{equation}
Substituting \eqref{beta*} into \eqref{likelihood}, we obtain
\begin{equation}\label{append_A_result}
\left\|
\bar {\mathbf y}-{\beta}^*_s(\theta)\operatorname{vec}(\mathbf q(\theta))
\right\|^2
=
\|\bar{\mathbf y}\|^2
-
\frac{
\left|
\mathbf a_r^H(\theta)\mathbf Y_s\mathbf X_s^H\mathbf b(\theta)
\right|^2
}{
\mathbf b^H(\theta)\mathbf X_s\mathbf X_s^H\mathbf b(\theta)
}.
\end{equation}
According to the result in \eqref{append_A_result}, the minimization problem in \eqref{likelihood} is equivalent to the maximization problem in \eqref{MLE}.
Thus, we can estimate the angle of the eavesdropper by solving the problem \eqref{MLE}.

\section{Proof of Lemma 2}
Since our focus is on evaluating the performance of the angle estimation, we define $\bm \beta =[\operatorname{Re}(\beta_s),\operatorname{Im}(\beta_s)]^T$ for simplicity.
Then, we denote $\bm J \in \mathbb{R}^{3\times3}$ as the FIM with regard to the estimated parameters $\bm\chi \triangleq [\theta_e,\bm\beta^T]^T$,
which is given by
\begin{equation}
\mathbf J
=
\begin{bmatrix}
J_{\theta_e\theta_e} & \mathbf J_{\theta_e\bm\beta} \\
\mathbf J_{\theta_e\bm\beta}^T & \mathbf J_{\bm\beta\bm\beta}
\end{bmatrix}.
\end{equation}
Each entry of $\bm J$ can be calculated using
\begin{equation} \label{J}
\bm J_{i,j}
=
2\operatorname{Re}\!\left\{
\frac{\partial (\beta_s\mathbf q(\theta_e))^H}{\partial \bm \chi_i}
\mathbf R_n^{-1}
\frac{\partial (\beta_s\mathbf q(\theta_e))}{\partial \bm\chi_j}
\right\},\forall i,j,
\end{equation}
where the noise covariance \(\mathbf R_n=\sigma_s^2\mathbf I_{N_rL}\). 
According to the definition of FIM, the CRB for angle estimation is given by
\begin{equation} \label{CRB_append}
\mathrm{CRB}(\theta_e)
=
[\mathbf J^{-1}]_{1,1}
=
[
J_{\theta_e\theta_e}
-
\mathbf J_{\theta_e\bm\beta}
\mathbf J_{\bm\beta\bm\beta}^{-1}
\mathbf J_{\theta_e\bm\beta}^T
]^{-1},
\end{equation}
where each entry of $\bm J$ can be calculated as follows referring to \eqref{J}.
Since
\[
\frac{\partial (\beta_s\mathbf q)}{\partial \theta_e}
=
\beta_s\dot{\mathbf q},
\qquad
\frac{\partial (\beta_s\mathbf q)}{\partial \operatorname{Re}\{\beta_s\}}
=
\mathbf q,
\qquad
\frac{\partial (\beta_s\mathbf q)}{\partial \operatorname{Im}\{\beta_s\}}
=
j\mathbf q,
\]

we obtain
\begin{align} \label{JJ}
J_{\theta_e\theta_e} &= \frac{2|\beta_s|^2}{\sigma_s^2}\dot{\mathbf q}^H\dot{\mathbf q},\notag  \\
\mathbf J_{\theta_e\bm\beta}
&=
\frac{2}{\sigma_s^2}
\begin{bmatrix}
\operatorname{Re}\{\beta_s^*\dot{\mathbf q}^H\mathbf q\}-\operatorname{Im}\{\beta_s^*\dot{\mathbf q}^H\mathbf q\}
\end{bmatrix},\notag\\
\mathbf J_{\bm\beta\bm\beta}
&=
\frac{2}{\sigma_s^2}
(
\mathbf q^H\mathbf q
)\mathbf I_2.
\end{align}
Substituting \eqref{JJ} into \eqref{CRB_append}, we have
\begin{equation}
\mathrm{CRB}(\theta_e)
=
\frac{\sigma_s^2}{
2|\beta_s|^2
(
\dot{\mathbf q}^H\dot{\mathbf q}
-
\frac{|\dot{\mathbf q}^H\mathbf q|^2}{\mathbf q^H\mathbf q}
)
},
\end{equation}
where
$
\dot{\mathbf q}(\theta_e)
=
\operatorname{vec}(
\dot{\mathbf a}_r(\theta_e)\mathbf b^H(\theta_e)\mathbf X_s
+
\mathbf a_r(\theta_e)\dot{\mathbf b}^H(\theta_e)\mathbf X_s
)
$, $
\dot{\mathbf a}_r(\theta_e)
\triangleq
\frac{\partial \mathbf a_r(\theta_e)}{\partial \theta_e}$
and $
\dot{\mathbf b}(\theta_e)
\triangleq
\frac{\partial \mathbf b(\theta_e)}{\partial \theta_e}.
$

Based on the expression in \eqref{Rx} and \(\|\mathbf a_r(\theta_e)\|^2=1\), the CRB for the angle $\theta_e$ can be further simplified as \eqref{CRB}.
\section{Proof of Lemma 3}
Let \(c_{\min}=\min_{k=1,\ldots,K}c_k\). Since \(c_k\ge c_{\min}\)
for all \(k\) and \(\beta>0\), we have
\[
\exp(-\beta c_k)\le \exp(-\beta c_{\min}), \quad \forall k .
\]
Moreover, there exists at least one index \(k^\star\) such that
\(c_{k^\star}=c_{\min}\). Therefore,
\[
\exp(-\beta c_{\min})
\le
\sum_{k=1}^{K}\exp(-\beta c_k)
\le
K\exp(-\beta c_{\min}).
\]
Taking the logarithm and multiplying by \(-1/\beta\) yield
\[
c_{\min}-\frac{1}{\beta}\log K
\le
-\frac{1}{\beta}
\log\sum_{k=1}^{K}\exp(-\beta c_k)
\le
c_{\min}.
\]
Substituting \(c_{\min}=\min_{k=1,\ldots,K}c_k\), we obtain
\[
\min_{k=1,\ldots,K} c_k-\frac{1}{\beta}\log K
\le
-\frac{1}{\beta}
\log\sum_{k=1}^{K}\exp(-\beta c_k)
\le
\min_{k=1,\ldots,K} c_k.
\]
The approximation gap is upper-bounded by \(\log K/\beta\), which
vanishes as \(\beta\to\infty\). This completes the proof.

\end{document}